\documentclass[10pt,oneside,final,twocolumn]{IEEEtran}
\usepackage[cmex10]{amsmath}
\usepackage{amsthm,amsfonts,amssymb}
\usepackage{graphicx,cite,calc}
\usepackage{xcolor}
\usepackage{relsize}

\newcommand{\norm}[1]{\left\lVert#1\right\rVert}

\usepackage{times}
\usepackage{booktabs}
\usepackage{balance}
\usepackage{multirow}
\usepackage{epsfig}
\usepackage{caption,subcaption}
\usepackage{hyperref}
\usepackage{color}
\usepackage{textcomp}
\usepackage{listings}
\usepackage{url}
\usepackage{pifont}
\usepackage{multicol}
\usepackage{graphicx}
\usepackage{tabu}
\usepackage{tikz}
\usepackage{xcolor}
\usepackage{colortbl}
\usepackage[export]{adjustbox}
\usepackage{array}
\usepackage[font=small,skip=0pt]{caption}
\usepackage[ruled,vlined]{algorithm2e}
\usepackage{wrapfig}
\usepackage{etoolbox}
\usepackage{etoolbox}
\preto\subequations{\ifhmode\unskip\fi}
\makeatletter
\newcommand*{\rom}[1]{\expandafter\@slowromancap\romannumeral #1@}
\makeatother

\newtheorem{thm}{Theorem}
\newtheorem{lem}[thm]{Lemma}
\newtheorem{remark}{Remark}

\usepackage{multirow}
\begin{document}
\bstctlcite{IEEEexample:BSTcontrol}
\title{Joint Network Lifetime Maximization and Relay Selection Design in Underwater Acoustic Sensor Networks}

\author{Z. Mohammadi, M. Soleimanpour-Moghadam, S. Talebi, H. Ahmadi
\thanks{Z. Mohammadi is 
with the Department of Electrical Engineering, Shahid Bahonar University of Kerman, Kerman, Iran (e-mail: zahramohammadi@eng.uk.ac.ir). 
M. Soleimanpour-moghadam is with the Texas Biomedical Research Institute, Texas, US (e-mail: 
msoleimanpour@txbiomed.org).
S. Talebi is with the Department of Electrical Engineering, Shahid Bahonar University of Kerman, Kerman, Iran and the Advanced Communications Research Institute, Sharif University, Tehran, Iran (e-mail: siamak.talebi@uk.ac.ir). H. Ahmadi is with the Department of Electrical Engineering, University of York, York, UK (e-mail: Hamed.ahmadi@york.ac.uk). \\Corresponding author: Mohadeseh Soleimanpour-Moghadam}
}

\maketitle
\begin{abstract}
The paper proposes a new approach to minimize the number of relays while maximizing the lifetime of underwater acoustic sensor networks (UASNs). This involves formulating the relay node placement (RNP) problem as a multi-objective optimization problem and employing the multi-objective lexicographic method (MOLM) to solve it. To achieve the optimal solution, the MOLM consists of two steps. First, the problem of lifetime maximization is tackled to find RNP solutions. This transforms the RNP into a non-convex optimization problem which is then converted into a convex programming equivalent. The proposed method has the same computational complexity as previous relay-node adjustment (RA) and difference convex algorithm (DCA) methods. The second step introduces a novel relay node selection to reach the optimal number of relays. Simulation results demonstrate that it has superior network lifetime and efficiency compared to RA and DCA.
\end{abstract}


\IEEEpeerreviewmaketitle
\section{Introduction}
\IEEEPARstart{U}nderwater acoustic sensor networks (UASNs) have attracted a great deal of attention for various applications including off-shore oil and gas extraction, oil spills, and natural calamities like tsunami and hurricane forecasts \cite{2020survey}.
These networks are composed of multiple nodes that use acoustic transceivers, which are more suitable for underwater environments than electromagnetic ones as acoustic signals experience less attenuation in the water. Each node collects data from its surroundings and transmits it to a surface buoy (SB). Unfortunately, because of the sparse deployment of UASNs, energy consumption is high, particularly for the nodes closest to the SB that send a large amount of data. This can lead to an {\it energy hole}, where the nodes closest to the SB die before the other nodes, making it impossible to forward the remaining data to the SB \cite{wadaa2005training}. To prolong the lifetime of the network and avoid the energy hole, much research has been conducted to reduce energy consumption in recent years.\par
\subsection{Literature review}

The work in \cite{wang2015energy} studied the cluster-based data forwarding to deal with the energy efficacy in UASNs. Based on this approach, the sensor nodes are classified into some clusters, and the cluster heads (CHs) are responsible for gathering data. Also, the residual energy and location of sensor nodes are taken into account to select the optimal CHs to prevent the energy hole problem.  
Using an autonomous underwater vehicle (AUV) to collect data from sensor nodes is an another solution to reduce energy consumption in UASNs. For instance, \cite{zhuo2020auv}, \cite{yan2018energy} demonstrate that by employing an AUV, data can be collected from gateways and gateways can be rotated over time to balance energy consumption. However, the transmission delays caused by the AUV are very long, making it difficult to use in time-sensitive applications such as temperature and salinity evaluations for red tide forecasts \cite{javaid2015efficient}. 
Power control schemes can also be used to adjust the power level of sensors during communication depending on the channel status and network conditions \cite{power_control, power_control1}. Additionally, sleeping schemes can be implemented for sensor nodes \cite{mac1, mac2} to help save energy by keeping them in the sleeping mode for as long as possible. However, these techniques pose new challenges when dealing with new technologies such as energy harvesting. To this end, \cite{EH1, EH2} present methods for optimizing time to foster energy harvesting and prolong the lifetime of the network. But energy harvesting-based networks struggle with the unpredictability of the harvested energy. There has been research on designing an optimal routing protocol. For example, routing protocols can balance energy consumption among nodes by assigning more energy to nodes that have higher traffic loads \cite{ahmed2016optimized}, using residual energy as the routing criteria \cite{wadud2019energy}, and splitting the transmission range into different power levels \cite{felemban2015underwater}. 
\par

The high costs and intricate characteristics of underwater networks often result in their sparse deployment, leading to increased energy consumption. To overcome this limitation and extend the network's lifetime, some authors propose the utilization of micro-relay nodes. These relay nodes possess identical size and power supply as the sensor nodes. One area of focus to enhance the network's lifetime is the strategic placement of these relay nodes, which is known as relay node placement (RNP). This aspect has garnered significant attention among researchers.
For example, the study by Das et al. \cite{das2017enhancement} proposed the utilization of relay nodes that dynamically move between communicating nodes, acting as intermediaries. This strategy aims to decrease the communication distance between nodes, thereby improving overall energy efficiency.
The authors in \cite{liu2017optimal} proposed a heuristic relay-node adjustment (RA) scheme for positioning relay nodes in 3-dimensional UASN. This scheme consists of two steps: initially randomizing the relay nodes on the water surface and then adjusting their depth.
However, the heuristic approach employed in \cite{liu2017optimal} may yield suboptimal solutions due to its reliance on simplified rules and assumptions. In contrast, our mathematical solution, as presented in \cite{mohammadi2018new, mohammadi2020increasing}, effectively overcomes these limitations, resulting in improved accuracy and efficiency.
Specifically, \cite{mohammadi2018new} introduced the concept of a line-segment relay node placement (LSRNP), which entailed positioning a relay node between a critical node and its farthest neighbor. However, it was demonstrated that LSRNP posed challenges as a complex, non-convex problem, making practical implementation difficult. Building upon this concept, \cite{mohammadi2020increasing} explored joint deployment of relay and sensor nodes, employing a difference convex algorithm (DCA) to derive a low-complexity solution.
In contrast to previous relay node placement schemes, which often employed simplified restrictions on the feasible space of relay nodes (i.e., a line between critical node and its farthest neighbor), our approach seeks to overcome this limitation by seeking the optimal position for relay nodes, ensuring the best possible arrangement. 
\par
\begin{table*}[!t]
\centering
\caption{Comparisons between related studies on energy management in UASNs}
\label{tab:summary_table}

\begin{tabular}{p{1.5cm}p{5cm} p{5cm} p{5cm}}
\hline
{Ref.} &{Methods Adoption} & {Problem Description} & {Performance Metrics}\\
\hline
\cite{wang2015energy} &Cluster-based data gathering & CH selection&Network lifetime\\
\cite{zhuo2020auv}, \cite{yan2018energy} & AUV-based data gathering & Path optimization of AUV& Collection delay\\
\cite{power_control, power_control1} &Adjusting the power level of sensors & Power allocation & Achievable throughput \\
\cite{mac1, mac2}& Changing the network topology&Link scheduling optimization & Energy consumption \\
\cite{EH1,EH2}& Energy harvesting& Optimizing time to harvest energy& Network lifetime\\
\cite{wadud2019energy}&Designing routing protocol& Optimizing the routing criteria& Balanced energy consumption\\
\cite{mohammadi2018new, mohammadi2020increasing} &Line-segment RNP & Optimizing the position of relay nodes & Network lifetime\\
{This work} & Multi-objective RNP& Network lifetime maximization and the number of relay nodes minimization& { Network lifetime and number of relay nodes}\\
\hline
\end{tabular}

\end{table*}
\subsection{Contribution}
The above mentioned RNP solutions contribute to extending the network lifetime but still fail to optimize the network lifetime. In this paper we introduce a framework that addresses the inherent shortcomings of the LSRNP. By exploring a broader range of potential relay node positions and leveraging advanced optimization techniques, we aim to significantly improve the accuracy and optimality of RNP.
Additionally, the number of relay nodes deployed in the RNP is an important factor to consider which is not considered in all previous works. In which the research in \cite{mohammadi2022modified} has investigated the effect of the number of relay nodes on the LSRNP and found that redundant relay nodes do not increase the network lifetime, but too many can lead to a decrease in network performance due to the extra time and energy spent in packet forwarding and receiving data. 
Motivated by these, we propose a framework that can solve the RNP efficiently without placing redundant relay nodes in the network. 
To the authors's best knowledge, while RNP is becoming a crucial topic in UASNs, this is one of the first times that network lifetime maximization is combined with number of relay nodes minimization. 
Overall, the paper's core contributions are described in the following.
\begin{enumerate}
\item We formulate the joint problem of maximizing network lifetime and minimizing the number of relay nodes (NLMA-RNMI) in relay-assisted UASNs as a multi-objective optimization problem (MOOP). We employ the multi-objective lexicographic method (MOLM) to achieve Pareto optimality for the given NLMA-RNMI MOOP. Based on the MOLM, the objectives are prioritized in order of importance, with network lifetime given the highest priority. Thus, our first aim is to maximize the network lifetime and subsequently minimize the number of relay nodes to prevent resource wastage.
\item To ensure that the network lifetime is maximized, we propose the Optimal Relay Node Setting (ORNS) algorithm, which models RNP as a mathematical optimization problem. This problem considers two criteria in the maximization process: balancing the energy consumption between sensor nodes and preventing the relay nodes positioned in outlier positions.
We have taken this into account in our problem formulation, leading to a MOOP RNP. To perform the proposed MOOP RNP, we use the $\epsilon$-constraint approach. By employing this approach, we prioritize the network lifetime of the critical node as the objective function to address the energy hole in the network. Additionally, we incorporate the relay node's lifetime as a constraint to ensure an improvement in the overall network lifetime. We demonstrate that the resulting RNP is non-convex and introduce a transformation to convert it into a convex optimization problem.
\item We then formulate a mixed-integer convex programming model to obtain the optimal number of relay nodes. To ensure a strong coupling relationship between network lifetime and the number of relay nodes, we consider the optimal value of network lifetime as a constraint in the problem of minimizing relay nodes. This ensures direct communication without the assistance of intermediate relay nodes when the energy consumption is low. Therefore, this approach is referred to as the relay selection design.
\end{enumerate}
To better highlight the contribution of this paper, Table \ref{tab:summary_table} presents a comparison of this work with different works in the literature.
We evaluate the time complexity of our RNP method, as well as the DCA and RA approaches. The results indicate that these approaches exhibit the same order of complexity. Additionally, we conduct several comprehensive simulations assess their performance. Through these simulations, we compare the effectiveness of our RNP approach with the others. The outcomes demonstrate that our approach outperforms the alternative methods in terms of network lifetime.

\par
The rest of the paper is organized as follows: In section \rom{2}, the model of the system and problem definition are illustrated and described. 
Some basics and preliminaries on the MOOP, MOLM, convex optimization problems, difference convex functions, and other basics \textemdash that are used 
in this work\textemdash are described in section \rom{3}. 
The proposed method to RNP and its equivalent convex programming scheme as well as the relay node selection scheme are presented in detail in section \rom{4}. In section \rom{5}, the complexity of the proposed RNP is evaluated.
Simulation and comparison results are provided in section \rom{6} and finally, the paper concludes in section \rom{7}, 
recapping its contribution.\\
Lightface letters denotes scalars. Boldface lowercase letters are employed to denote vectors and boldface uppercase letters to denote matrices. The operations $\rm{ E (.)}, (.)^T$ and $\|.\|_p$ denote the expectation operator, transpose, and p-norm respectively. 
\textcolor{black}{The $[{\mathbf A}]_{n*}$ and ${\mathbf a} [n]$ stands for $n$-th row and $n$-th column of matrix $\mathbf A$ respectively,} and to denote the $(m,n)$ entry of matrix $\mathbf A$ we use $[{\mathbf A}]_{mn}$. 
The $\mathbf{1}$ refers to a vector with all elements equal to one while $\mathbf{e}_i$ denotes a unit vector with element $i$ equals one and zeros everywhere else. 
We also denote the size, i.e., cardinality, of set $S$ by $|S|$. 
$\rm {dom}$ is an abbreviation of domain where $\rm {dom} (g)$ describes all the values that go into the function $g$.
\section{System model and problem formulation}{\label{sec2}}
Consider a 3-dimensional multi-hop UASN, as depicted in Fig \ref{Fig1}, which consists of $\rm N$ sensor nodes and $\rm M$ relay nodes deployed in a designated search field to gather data about the environment. The sensor nodes are constrained by limited battery power and their positions are randomly determined. The relay nodes, on the other hand, are strategically placed to maximize the network lifetime, but are not capable of sensing information from the environment. Additionally, the SB is situated on the ocean surface to communicate with a satellite that forwards the data gathered by the sensor nodes to the onshore sink. The SB is responsible for making RNP decisions in the field and all the sensor nodes are required to communicate with it in order to forward parameters. This scheme is effective in maximizing the network lifetime, while minimizing the number of relay nodes.
\begin{figure}[!b]
\centering{\includegraphics[scale=.26]{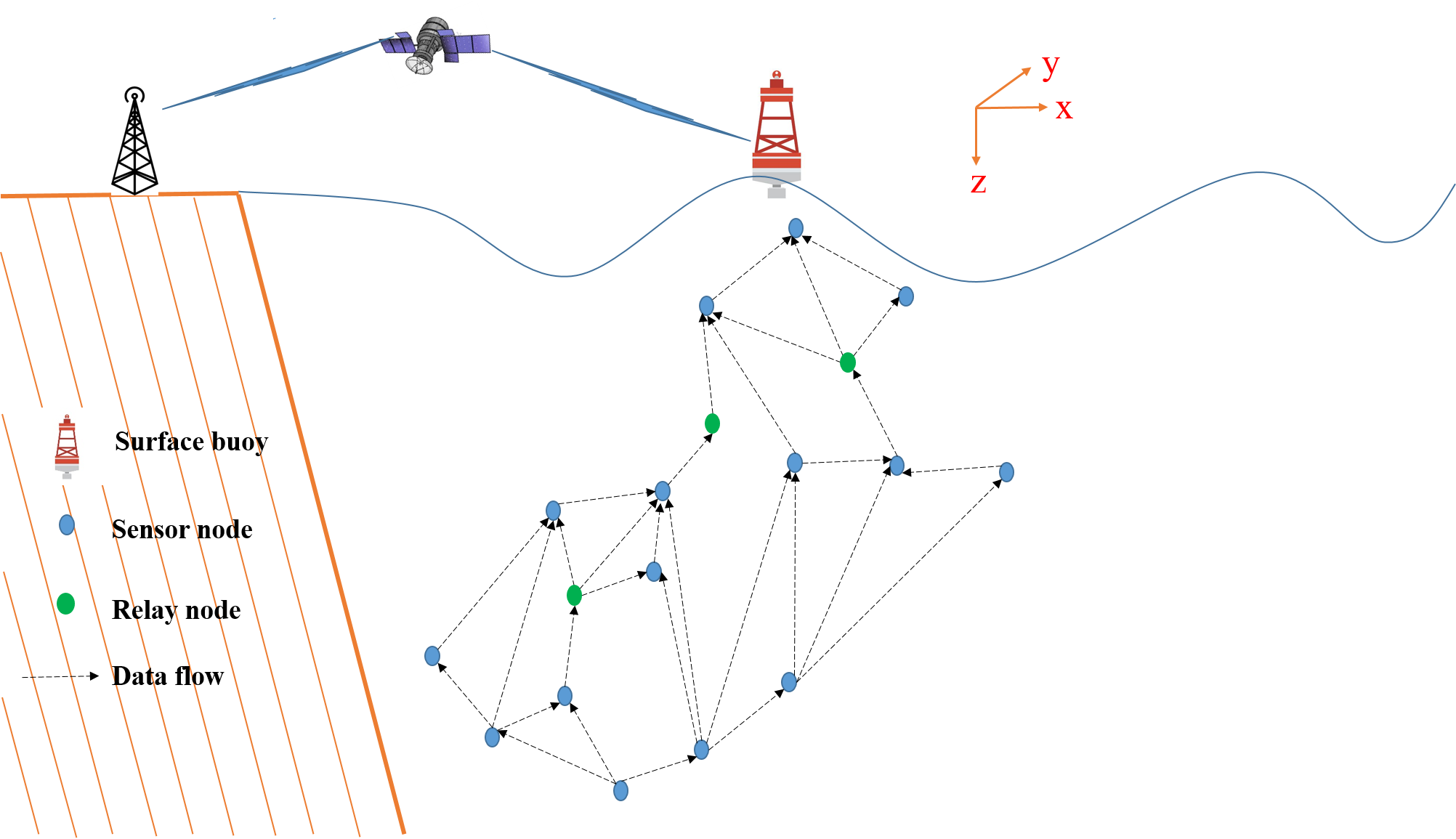}}
\captionsetup{justification=centering}
\caption{Network model with 3-dimensional Cartesian coordinates\label{Fig1}}
\captionsetup{justification=centering}
\label{Fig1}
\end{figure}
It should be noted that the origin of the Cartesian coordinates system is located on SB. In the proposed scheme, all nodes have the communication range $C_R$, so any two nodes out of this range will not be able to communicate together. 
The rate array $\mathbf{R}$ is given as 
\begin{equation}
\mathbf{R}=
\left[\mathbf{R}_{1*}^T, \ldots, \mathbf{R}_{i*}^T, \ldots, \mathbf{R}_{|\mathcal N|*}^T\right]^T,
\end{equation}
where $\mathcal N$ is the set of all nodes, including sensor nodes ($ \mathcal S$), relay nodes ($ \mathcal R$), and the SB, and $\mathbf{R}_{i*} \in \mathbb{Z}_+^{|\mathcal N|\times 1}$ is the outgoing flow vector from node $i$ to other nodes. Additionally, $\mathbf{r}\left[i\right] \in \mathbb{Z}_+^{|\mathcal N|\times 1}$ is the incoming flow vector from other nodes to node $n$ for $i=1,\ldots,|\mathcal N|$.
In order to ensure that the flow rate array meets certain specifications, there are several conditions that must be taken into account. 
\begin{enumerate}
\item 
At each sensor node $k$, the sum of outgoing flow rates must be equal to the sum of incoming flow rates and the generation rate, given as 
\begin{equation}
\mathbf{R}_{k*} \mathbf{1}=\mathbf{1}^T \mathbf{r}\left[k\right]+g_k,
\end{equation}
where $g_k$ is the generation rate of sensor node $k$. Additionally, at each relay node, the sum of outgoing flow rates should be equal to the sum of incoming flow rates, 
\begin{equation}\mathbf{R}_{i*} \mathbf{1}=\mathbf{1}^T \mathbf{r}\left[i\right].\end{equation}
\item Underwater acoustic communication faces challenges such as limited bandwidth, long propagation delays, multipath fading, and high signal attenuation. These factors affect the achievable data rate and overall link capacity in underwater acoustic sensor networks.
In the context of UASNs, the node's link capacity ($L_c$) refers to the maximum data rate or throughput that can be achieved over a communication link between two nodes in the network \cite{akyildiz2005underwater}.
Therefore, the sum of outgoing flow rates of each node must be less than the node's link capacity ($L_c$), 
\begin{equation}
\mathbf{R}_{n*} \mathbf{1} \leq L_c,\end{equation}
for each node $n\in \mathcal R \cup S$. 
\end{enumerate}
\subsection{Energy consumption model}

Based on the Urick model \cite{urick1975}, the energy consumption of sending one bit of data from one node, $i$, to another node, $j$ is expressed as \cite{wang2016energy}, \cite{liu2017optimal}
\begin{equation}\label{powerconsumption}
p_{ij}=
\begin{cases}
p_s +\epsilon_{f s}{d^2_{ij}} & d_{ij}<d _t \\
p_s +\epsilon_{mp}{d^4_{ij}} & d_{ij} \geq d _t
\end{cases},
\end{equation}
where $d_{ij}$ is the 3-dimensional Euclidean distance between nodes $i$ and $j$ given by $\norm {\mathbf l_i-\mathbf l_j}$, with $\mathbf l_i$ being the position vector of node $i$. Additionaly, $d_t$ is a threshold distance to transmit data; $p_s$ 
is the the power consumption for processing in sending data; $\epsilon_{fs}$ and $\epsilon_{mp}$ represent
transmit amplifier coefficient of free space and multipath
model, respectively. If $d_{ij}<d_t$, the amplifier coefficient of free space model $\epsilon_{f s}$
is adopted. Otherwise, the amplifier coefficient of free multipath model $\epsilon_{mp}$ is adopted.
In UASNs, no matter the free space module or multipath model,
the amplifier coefficient is defined as $\alpha (f)^{d_{ij}}$ \cite{wang2016energy}, \cite{liu2017optimal}, where $\alpha (f)$ is
the absorption coefficient which is derived from Throp's formula \cite{thorp1967analytic} as
$10{\rm log}\alpha (f)=0.1\frac{f^2}{1+f^2}+\frac{40f^2}{4100+f^2}+2.75\times10^{-4}f^2+0.003$ 
for frequencies of the acoustic signal above a few hundred Hertz and $f$ is the frequency of acoustic signal. 
Therefore, the expression (\ref{powerconsumption}) is rewritten as:
\begin{equation}
p_{ij}=
\begin{cases}
p_s +\alpha (f)^{d_{ij}}{d^2_{ij}} & d_{ij} <d _t \\
p_s +\alpha (f)^{d_{ij}}{d^4_{ij}} & d_{ij} \geq d _t
\end{cases},
\end{equation}

%
%

The lifetime of an underwater acoustic sensor network (UASN) is defined as the duration until the death of the first node, as effective communication is only possible until then \cite{cao2015balance}, \cite{akbar2016efficient}. The lifetime of a node, $i$, is expressed as the ratio of its residual energy, $\epsilon_i$, to its total energy consumption \cite{liu2017optimal}:
\begin{equation}\label{N_def}
\tau_i=\frac{\epsilon_i}{\sum_{j\in \mathcal N}^{i\neq j} p_{ij} [\mathbf{R}]_{ij} + p_r \sum_{k\in \mathcal S \cup \mathcal R}^{k\neq i}[\mathbf{R}]_{ki}},
\end{equation}
where $p_r$ is the energy consumption for receiving one bit, and $[\mathbf{R}]_{ij}$ is the outgoing flow from node $i$ to node $j$, and $[\mathbf{R}]_{ki}$ is the incoming flow from node $k$ to node $i$.

The following assumptions are made in the investigation: (1) All sensor and relay nodes have the same communication range; (2) the movement of sensors is predictable, and the position of the nodes is known through the localization process; (3) Link capacity (${L}_c$) is considered constant and equal for all sensor and relay nodes \cite{liu2017optimal}; and (4) Before placing relay nodes, the network is connected, meaning each sensor node has a route to the SB.
\subsection{Problem formulation} 
Our goal is to maximize the UASN lifetime while jointly minimizing the number of required relay nodes. 
The multi-objective optimization problem is formulated as below:
\begin{subequations}
\begin{align}
&\{\min {\rm M}, \max\{\min_{i \in \mathcal S \cup \mathcal R} \tau_i\}\},\\\intertext{s.t.}
&\tau_i=\frac{\epsilon_i}{\sum_{j\in \mathcal N}^{i\neq j} p_{ij} [\mathbf{R}]_{ij} + p_r \sum_{k\in \mathcal S \cup \mathcal R}^{k\neq i}[\mathbf{R}]_{ki}}, i \in \mathcal S \cup \mathcal R \label{eq14b}\\
&|\mathcal R| = {\rm M}\label{eq14c}\\
&p_{ij}=
\begin{cases}
p_s +\alpha^{d_{ij}}{d^2_{ij}} & d_{ij} <d _t \\
p_s + \alpha^{d_{ij}}{d^4_{ij}} & d_{ij} \geq d _t
\end{cases}, 
\label{eq14d}\\
&d_{ij}^2=\norm{\mathbf l_i-\mathbf l_j}^2, \forall i,j\in \mathcal N,\label{eq14e}\\
& \mathbf l_{r_i} \in X_c, \forall r_i\in \mathcal R. \label{eq14f}
\end{align}
\end{subequations}
Constraint (\ref{eq14c}) sets the number of relay nodes to $\rm M$. Constraint (\ref{eq14d}) and (\ref{eq14e}) calculate the energy consumption between node $i$ and $j$ based on the 3-dimensional Euclidean distance between them. Lastly, constraint (\ref{eq14f}) requires that relay nodes must be positioned within the cylindrical area of the surveillance field with radius $R_s$ and depth $H_s$.
\section{Preliminaries}{\label{sec3}}
In this section, we introduce preliminaries that we will use in the rest of our study.
\newtheorem{defn}{{\bf Definition}}
\begin{defn}\label{def2}
We consider a general MOOP as
\begin{subequations}
\begin{align}
&\min \{f_1(\mathbf x),...,f_K(\mathbf x)\}\\\intertext{s.t}
&\mathbf x\in \mathcal X 
\end{align}
\end{subequations}
with $K$ objectives and a feasible set $\mathcal X$.
Using MOLM, the objectives are ranked in order of importance from best to worst. The problem then begins with the most important objective and continues with the objectives in the order of their importance. Specifically, in step $i$, $f_i^*$ is obtained by sequentially minimizing the objective $f_i$. It is worth noting that, the computed optimal value of each objective is added as a constraint for the subsequent optimization steps.
\end{defn}
\begin{defn}\label{def2}
In a convex optimization problem, we minimize a convex objective function over a convex set. This problem is of the form
\begin{subequations}
\begin{align}
&\min f_0(\mathbf{x})\\\intertext{s.t.}
&f_i(\mathbf{x})\leq 0, i=1, \ldots, m\\
&h_i(\mathbf{x})=0, i=1, \ldots, p
\end{align}
\end{subequations}
where $\mathbf{x}\in \mathbb{R}^n$ and $f_0, \dots, f_m: \mathbb{R}^n\rightarrow\mathbb{R}$ are convex functions and $h_i(\mathbf{x}): \mathbb{R}^n\rightarrow\mathbb{R}$ are affine functions \cite{boyd2004convex}. The important property of this problem is that any locally optimal solution is also globally optimal.
\end{defn}
\begin{defn}
Let $\Omega$ be a convex set in $\mathbb{R}^n$. We say that a function is a difference convex (DC) if it can be expressed as the difference of two convex functions on $\Omega$, i.e. if $f(x)=f_1(\mathbf{x})-f_2(\mathbf{x})$, where $f_1$ and
$f_2$ are convex functions on $\Omega$ \cite{tuy1998convex}. 
The function $f(x)$ 
is convex when $f_1(\mathbf{x})$ and $f_2(\mathbf{x})$ are convex and affine functions, respectively. 
In general, each convex function is a DC function, but its reverse is not true.
\end{defn}
\begin{lem}\label{lem1}
The inverse of the function $f: \mathbb{R}^n \rightarrow \mathbb{R}$ will reach its minimum point at $\mathbf{x}_0 \in \mathbb{R}^n$ if $\mathbf{x}_0$ is the maximum point of $f$. Additionally, this result also holds for the converse.
\end{lem}
\begin{proof}
See the theorem 1. 46 in \cite{silverman1989essential}.
\end{proof}
\section{Proposed method}
Using the lexicographic optimization algorithm, the lifetime of a network and the number of relay nodes are optimized in a hierarchical manner. This process involves two steps to arrive at the optimal solution.
\subsection{ Step 1 - Network lifetime maximization: ORNS scheme}
In this case, the problem can be defined as 
\begin{subequations}
\begin{align}
&\max\{\min_{i \in \mathcal S \cup \mathcal R} \tau_i\}\\\intertext{s.t.}
&\tau_i=\frac{\epsilon_i}{\sum_{j\in \mathcal N}^{i\neq j} p_{ij} [\mathbf{R}]_{ij} + p_r \sum_{k\in \mathcal S \cup \mathcal R}^{k\neq i}[\mathbf{R}]_{ki}}, i \in \mathcal S \cup \mathcal R, \\
& |\mathcal M| ={\rm{M_0}} \\ 
&p_{ij}=
\begin{cases}
p_s +\alpha^{d_{ij}}{d^2_{ij}} & d_{ij} <d _t \\
p_s + \alpha^{d_{ij}}{d^4_{ij}} & d_{ij} \geq d _t
\end{cases}, \\
&d_{ij}^2=\norm{\mathbf l_i-\mathbf l_j}^2, \forall i,j\in \mathcal N,\\
& \mathbf l_{r_i} \in X_c, \forall r_i\in \mathcal R. 
\end{align}
\end{subequations}
The problem at hand is a challenging, non-convex NP-hard problem, which we have attempted to address through the use of ORNS. This method is designed to determine the ideal location of a single relay node ($r_i$) and can be extended to identify the best coordinates of all the relay nodes. Furthermore, ORNS calculates which nodes ($N_ {c_i}^u$) receive data from the most critical node ($c_i$) of the network. This set can be expressed as: 
\begin{equation}
N_ {c_i}^u=\{j;[\mathbf{R}]_{{c_i}j} >0\},\label{eq22}
\end{equation}
where $j$ is the nodes that is connected to $c_i$. With the optimal coordinates of $r_i$ determined, it can then act as a router to facilitate the connection between $c_i$ and the elements of $N_ {c_i}^u$. 
In other words, after placing $r_i$ in the network the following equations can be written:
\begin{equation}
[\mathbf{R}]_{c_{i}r_{i}}=\sum_{j\in N_ {c_i}^u} [\mathbf{R}]_{c_{i}j},\label{eq23}
\end{equation}
\begin{equation}
[\mathbf{R}]_{r_{i}j}= [\mathbf{R}]_{c_{i}j}, \forall j \in N_ {c_i}^u.\label{eq24}
\end{equation}
According to the above equations, $r_i$ acts as an intermediate node to forward the information of node $c_i$ and thus:
\begin{equation}
[\mathbf{R}]_{c_{i}j}=0, \forall j \in N_ {c_i}^u.\label{eq25}
\end{equation}
From the network lifetime perspective, mathematically, two remarks must be made in regard to the RNP: 
\begin{remark}
In a valid RNP, less energy should be consumed than in a direct transmission, and the network lifetime is maximized when the $c_i$ lifetime is maximized, which can be expressed as
\begin{equation}
\mathbf{l}_{{r}_{{i}}}=\mathrm{argmax}\left( {\tau}_{{c}_{{i}}} \right).
\end{equation}
\end{remark}
\begin{remark}
By defining the residual energy factor of $c_{i}$
as $\rm {RF}_i \triangleq \frac{\epsilon_{c_i}}{\epsilon_{p}}, $
where $\epsilon_{p}$ is the primary energy
of $c_{i}$, the efficient RNP (ERNP) will take into account the
lifetime of relays especially when the $\rm {RF}_i $ is high, which can be expressed as:
\begin{equation}
\mathbf{l}_{{r}_{ {i}}}=\mathrm{argmax} \left( \tau_{{r}_{i}} \right)\ \ {s.t.\ }p_{r_{i}j} < p_{c_{i}j}
\end{equation}
\end{remark}
To further illustrate this concept, an example of an RNP in two system scenarios is given in Fig. \ref {Fig2} where the ERNP considers the case when the relay node is located within the convex hull of $c_i$ and its upper neighbors $\left( N_{c_{i}}^{u} \right)$, mathematically expressed as follows:
\begin{equation}
\mathbf{l}_{r_{i}} \in {\mathrm{conv}}\left\{ \mathbf{l}_{j},\ j \in N_{c_{i}}^{u}\cup c_i \right\} \triangleq \mathbf{\mathrm{conv}}_{{i}},
\label{eq28}
\end{equation}
where $\mathrm{conv}_{i}$ is a convex combination of these $\left| N_{c_{i}}^{u} \right| + 1$ vectors. This system is true if and only if there is a solution to the following system: 
\begin{equation}
\mathbf{l}_{r_{i}} =\Big \{ \sum_{ j \in N_{c_{i}}^{u}\cup c_i }^{}{\theta_{j}\mathbf{l}_{j}}| \theta_{j} \geq 0,\sum_{ j \in N_{c_{i}}^{u}\cup c_i }^{}\theta_{j} = 1\Big\}.\label{cvx_def1}
\end{equation}
\begin{figure}[!t]
\centering{\includegraphics[width=.9\columnwidth]{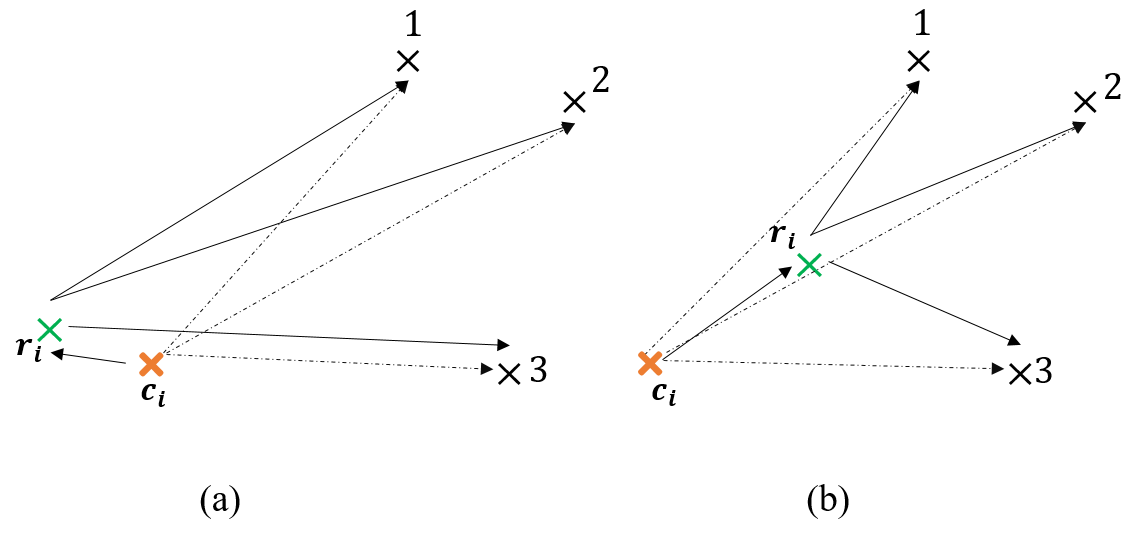}}
\caption{Routing traffic by the two RNPs when $|N_{c_{i}}^{u}|=3$, (a) $p_{r_ij}>p_{c_ij}$, (b) $p_{r_ij}<p_{c_ij}$ (ERNP)}
\label{Fig2}
\end{figure}
\begin{figure}[!t]
\centering{\includegraphics[width=.7\columnwidth]{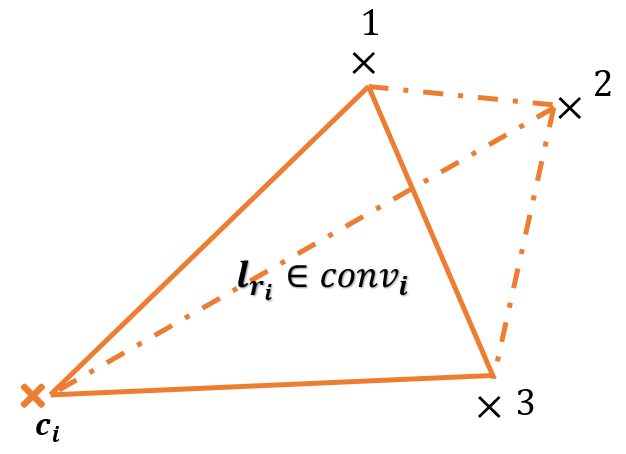}}
\caption{Construction of the feasible space to locate $r_i$ using convex hull when $|N_{c_{i}}^{u}|=3$}
\label{Fig3}
\end{figure}
{\color{black}{
The feasible space created by (\ref{cvx_def1}) for \(r_{i}\) is shown in Fig. \ref{Fig3}. \par
In summary, to meet all the above conditions, the problem of network lifetime maximization can be given as 
\begin{subequations}\label{p31}
\begin{align}
&\mathbf{y}_{i} = \arg\left( \max\tau_{c_{i}},\max\tau_{r_{i}} \right),\\
&s.t.\nonumber\\
& \tau_{c_{i}} = \frac{\epsilon_{c_{i}}}{p_{c_{i}r_{i}} [\mathbf{R}]_{{c_i}{r_i}} + p_{\text{r}}\sum_{ k \in \mathcal S \cup \mathcal R }^{} [\mathbf{R}]_{k{c_i}} }, \\
& \tau_{r_{i}} = \frac{\epsilon_{r_{i}}}{\sum_{j \in N_{c_{i}}^{u}}^{}{p_{r_{i}j} [\mathbf{R}]_{{r_i}j} } 
+ p_{{r}} [\mathbf{R}]_{{c_i}{r_i}}} ,\\
&p_{c_{i}r_{i}} =
\begin{cases}
p_{s} + \alpha\left( f \right)^{d_{c_{i}r_{i}}}d_{c_{i}r_{i}}^{2}\ ,d_{c_{i}r_{i}} < d_{t}\\
p_{s} + \alpha\left( f \right)^{d_{c_{i}r_{i}}}d_{c_{i}r_{i}}^{4},\ d_{c_{i}r_{i}} \geq d_{t} 
\end{cases},
\\
&d_{c_{i}r_{i}} ^2= \left\| \mathbf{l}_{r_{i}} - \mathbf{l}_{c_{i}} \right\|^{2},\\
&p_{r_{i}j} =
\begin{cases}
p_{s} + \alpha\left( f \right)^{d_{r_{i}j}}d_{r_{i}j}^{2}\ ,d_{r_{i}j} < d_{t} \\
p_{s} + \alpha\left( f \right)^{d_{r_{i}j}}d_{r_{i}j}^{4},\ d_{r_{i}j} \geq d_{t}
\end{cases},\forall j \in N_{c_{i}}^{u},\\
&d_{r_{i}j}^{2} = \left\| \mathbf{l}_{r_{i}} - \mathbf{l}_{j} \right\|^{2},\ \forall j \in N_{c_{i}}^{u},\\
&\mathbf{l}_{r_{i}} \in \mathbf{\rm{conv}}_{{i}},\\
&\left\{ p_{r_{i}j},d_{r_{i}j},p_{c_{i}r_{i}},d_{c_{i}r_{i}} \right\} \in \mathbb{R}_{+}^{\ },\ \forall j,
\end{align}
\end{subequations}
that $\mathbf{y}_{{i}}$ denotes the decision variable where 
$
\mathbf{y}_{{i}} \triangleq \left\lbrack \mathbf{p}_{r_i},\mathbf{d}_{r_i},p_{c_{i}r_{i}},d_{c_{i}r_{i}},\mathbf{l}_{r_{i}},\mathbf{\Theta} \right\rbrack \in \mathbb{R}^{3\left| N_{c_{i}}^{u} \right| + 6},
$
$
\mathbf{\Theta} = \left\lbrack \theta_{1}, \ldots ,\theta_{\left| N_{c_{i}}^{u} \right|+1} \right\rbrack^{T}, 
$
$
\mathbf{p}_{r_{i}} = \left\lbrack p_{r_{i}1}, \ldots, p_{r_{i}\left| N_{c_{i}}^{u} \right|} \right\rbrack^{T},
$
and 
$
\mathbf{d}_{r_{i}} = \left\lbrack d_{r_{i}1}, \ldots, d_{r_{i}\left| N_{c_{i}}^{u} \right|} \right\rbrack^{T}.
$

It is important to note that solely maximizing the lifetime of the critical node can solve the problem of energy hole among sensors and does not guarantee improving the entire network lifetime. Therefore, our proposed method takes into account the maximization of both relay nodes and critical nodes simultaneously. To achieve this, we employ multi-objective optimization techniques to identify optimal values by exploring trade-offs and finding desirable solutions across multiple objectives. 
In the following, we utilize the $\epsilon$-constraint \footnote{Based on the $\epsilon$-constraint method, a main and a secondary objective functions are selected, and the purpose is to optimize the main objective function and limit the secondary function by some allowable amount. } method to solve the problem (\ref{p31}). Based on this approach, the lifetime of critical nodes serves as the main objective function, and the lifetime of the relay node is the secondary objective. 
The relay node's lifetime, as the secondary objective, must be greater than that of the critical node and is added as a constraint in the optimization problem. Indeed, we incorporate relay node's lifetime as a constraint in the optimization problem to ensure that the entire the network lifetime is improved. 
By employing this method, we guarantee the resolving the energy hole by maximization the lifetime of critical nodes while preserving the lifetime of relay nodes and avoiding their placement in outlying positions. This approach effectively resolves the energy hole problem caused between sensor nodes and enhances the overall network lifetime. To sum up, the RNP is given as
\begin{subequations}\label{p32}
\begin{align}
&\mathbf{x}_{{i}} = \arg \max\tau_{c_{i}},\\
&s.t.\nonumber\\
& \tau_{c_{i}} = \frac{\epsilon_{c_{i}}}{p_{c_{i}r_{i}} [\mathbf{R}]_{{c_i}{r_i}} + p_{\text{r}}\sum_{ k \in \mathcal S \cup \mathcal R }^{} [\mathbf{R}]_{k{c_i}} }, \\
& \tau_{r_{i}} \geq \tau_{c_{i}}, \label{eq34c} \\
&p_{c_{i}r_{i}} =
\begin{cases}
p_{s} + \alpha\left( f \right)^{d_{c_{i}r_{i}}}d_{c_{i}r_{i}}^{2}\ ,d_{c_{i}r_{i}} < d_{t}\\
p_{s} + \alpha\left( f \right)^{d_{c_{i}r_{i}}}d_{c_{i}r_{i}}^{4},\ d_{c_{i}r_{i}} \geq d_{t} 
\end{cases},
\\
&d_{c_{i}r_{i}} ^2= \left\| \mathbf{l}_{r_{i}} - \mathbf{l}_{c_{i}} \right\|^{2},\label{32_e}\\
&\mathbf{l}_{r_{i}} = \sum_{j \in N_{c_{i}}^{u}\cup c_i }^{}{\theta_{j}\mathbf{l}_{j}},\\
&\mathbf{\Theta} \succcurlyeq \mathbf{0},\\
&\mathbf{1}^T \mathbf{\Theta} = 1,\\ 
&\mathbf{p}_{r_i} \succcurlyeq \mathbf{0},\\
&\mathbf{d}_{r_i} \succcurlyeq \mathbf{0}.
\end{align}
\end{subequations}
where the decision variable is updated as
$\mathbf{x}_{{i}} \triangleq \left\lbrack \mathbf{p}_{r_i},p_{c_{i}r_{i}},d_{c_{i}r_{i}},\mathbf{l}_{r_{i}},\mathbf{\Theta} \right\rbrack \in \mathbb{R}^{2\left| N_{c_{i}}^{u} \right| + 6}. $
It can be shown that the (\ref{eq34c}) is rewritten as a linear constraint
\begin{equation}
\mathbf{1}^{T}\mathbf{D}\mathbf{x}_{{i}} \leq \gamma_{0},
\end{equation}
where
\begin{equation}
\gamma_{0} = \left( p_{{r}}\sum_{ k \in \mathcal S \cup \mathcal R }^{} [\mathbf{R}]_{k{c_i}} \right)\left( \frac{\epsilon_{r_{i}}}{\epsilon_{c_{i}}} \right) - p_{\text{r}} [\mathbf{R}]_{{c_i}{r_i}},
\end{equation}
\begin{equation}
\mathbf{D} = \mathrm{diag}\ \left( a_{1},\ldots ,a_{2\left| N_{c_{i}}^{u} \right| + 6}\right),
\end{equation}
and 
\begin{equation}
a_{j} = 
\begin{cases}
[\mathbf{R}]_{{r_i}j},\ j \in \{ 1,\ldots,\left| N_{c_{i}}^{u} \right|\} \\
- \left( \frac{\epsilon_{r_{i}}}{\epsilon_{c_{i}}} \right) [\mathbf{R}]_{{c_i}{r_i}},\ j = \left| N_{c_{i}}^{u} \right| + 1 \\
0,\ o.w \\
\end{cases}.
\end{equation}
Problem (\ref{p32}) is a non-convex problem that includes the convex objective function with convex (and/or linear) and DC constraints. 
Considering that the constraint (\ref{32_e}) is DC and non-convex, this leads to the non-convexity of the proposed RNP.
However, we explore a novel convex programming model equivalent to this problem.
To do so, we apply a novel transformation to these problems whose detailed expressions are developed in Appendix \ref{A1}. 
Based on the proposed transformation, lemma \ref{lem1}, and definition of variable $t$, the proposed problem to place $r_i$ in the network is the form of
{\color{black}{\begin{subequations}
\begin{align}
&[\mathbf{x}_{{i}}, t]= \arg \min p_{c_{i}r_{i}}\\
&s.t.\nonumber\\
&p_{c_{i}r_{i}}=
\begin{cases}
&p_{s} + \alpha\left( f \right)^{d_{c_{i}r_{i}}}d_{c_{i}r_{i}}^{2} \ ,d_{c_{i}r_{i}} < d_{t}\\
& p_{s} + \alpha\left( f \right)^{d_{c_{i}r_{i}}}d_{c_{i}r_{i}}^{4} ,\ d_{c_{i}r_{i}} \geq d_{t} \\
\end{cases}\\
&\mathbf{1}^{T}\mathbf{D}\mathbf{x}_{{i}} \leq \gamma_{0}, \\
& \left\| \mathbf{l}_{r_{i}} - \mathbf{l}_{c_{i}} \right\|^2-t = 0,\label{40e}\\
& d_{c_{i}r_{i}}^2 -t= 0\\
&\mathbf{l}_{r_{i}} = \sum_{j \in N_{c_{i}}^{u}\cup c_i }^{}{\theta_{j}\mathbf{l}_{j}},\\
&\mathbf{\Theta} \succcurlyeq \mathbf{0},\\
&\mathbf{1}^T \mathbf{\Theta} = 1,\\
&\mathbf{p}_{r_i} \succcurlyeq \mathbf{0},\\
&\mathbf{d}_{r_i} \succcurlyeq \mathbf{0}.
\end{align}
\label{p44}
\end{subequations}}}
\vspace{-1mm}
which is known as the epigraph form of the problem: 
{\color{black}{\begin{subequations}
\begin{align}
&[\mathbf{x}_{{i}}, t]= \arg \min \begin{cases}
&p_{s} + \alpha\left( f \right)^{d_{c_{i}r_{i}}}d_{c_{i}r_{i}}^{2} \ ,d_{c_{i}r_{i}} < d_{t}\\
& p_{s} + \alpha\left( f \right)^{d_{c_{i}r_{i}}}d_{c_{i}r_{i}}^{4} ,\ d_{c_{i}r_{i}} \geq d_{t} \\
\end{cases}\\
&s.t.\nonumber\\
&\mathbf{1}^{T}\mathbf{D}\mathbf{x}_{{i}} \leq \gamma_{0}, \\
& \left\| \mathbf{l}_{r_{i}} - \mathbf{l}_{c_{i}} \right\|^2-t \leq 0,\label{40e}\\
& d_{c_{i}r_{i}}^2 -t\leq 0\\
&\mathbf{l}_{r_{i}} = \sum_{j \in N_{c_{i}}^{u}\cup c_i }^{}{\theta_{j}\mathbf{l}_{j}},\\
&\mathbf{\Theta} \succcurlyeq \mathbf{0},\\
&\mathbf{1}^T \mathbf{\Theta} = 1,\\
&\mathbf{p}_{r_i} \succcurlyeq \mathbf{0},\\
&\mathbf{d}_{r_i} \succcurlyeq \mathbf{0}.
\end{align}
\label{p45}
\end{subequations}}}
Our proposed RNP belongs to non-differentiable convex optimization problems. 
The advantage of convex problems over non-convex counterparts is that, in general, a global optimum can be computed with good precision and within a reasonable time, independent of initialization \cite{soleimanpour2018low}. To obtain the optimal position of the relay node we resort to off-the-shelf convex solver CVX which is a Matlab-based modeling system for convex optimization. CVX can solve much more complex convex optimization problems, including non-differentiable functions. It was developed using interior point methods and gives numerical solutions for the convex optimization problem.}}
By solving it, a set of optimal solutions $(\mathbf l_{r_1},...,\mathbf l_{r_{M_0}})$ and $\mathbf{p}_r^*=[\mathbf 1^T\mathbf p_{r_1},...,\mathbf 1^T \mathbf p_{r_{M_0}}]^T$ showing the vector position of relay nodes and overall transmitting energy consumption of them is obtained, respectively. This optimal solution shows the optimal value $\tau^*$. The pseudo-code of the proposed network lifetime is given in algorithm 1.
\\
\begin{algorithm}

\KwData{The set of sensor nodes ($\mathcal S$)}
\KwResult{Position of relay nodes, rate array}
1: \textbf{for} each relay node $r_i$ \textbf{do}\\
2:\quad \textbf{for} each node $n \in \mathcal S\cup R$ \textbf{do} \\
3:\quad \quad Compute $\tau_n$ using (\ref{N_def})\\
4:\quad \textbf{end for}\\
5:\quad $c_i=\arg \min \tau_n$\\
6:\quad Construct the set $N_ {c_i}^u$ using (\ref{eq22})\\
7:\quad Define the convex hull system as given in (\ref{cvx_def1})\\
8:\quad Define the multi-objective problem (\ref{p31})\\
9:\quad Apply the $\epsilon$-constraint approach to convert the multi-objective problem (\ref{p31}) into the single-objective problem (\ref{p32})\\
10:\quad Apply the transformation (\ref{p49}) and Lemma \ref{lem1} to form problem (\ref{p44})\\
11:\quad Define the convex-based RNP (\ref{p45}) by using the epigraph form of (\ref{p44})\\
12:\quad Solve problem (\ref{p45}) using CVX tool to obtain ${\mathbf l}_{r_i}$\\
13:\quad Update $\mathbf R$ based on Eqs. (\ref{eq23})-(\ref{eq25})\\
14: \textbf{end for}\\
15: \textbf{return} $\mathbf R, {\mathbf l}_{r_i}, i=1, \ldots, \rm M$\\
\caption{Summary of proposed ORNS approach}

\end{algorithm}
\begin{table*}[!t]
\centering
\caption {\textcolor{black}{Computations of proposed ORNS method}}\label{t_complexity}

\resizebox{0.9\textwidth}{!}{
\begin{tabular}{|p{1.8cm} | p{3.4cm}|p{3.4cm}| p{6.5cm}|}
\hline
Operation &Computations of formulating (\ref{p45})& Computations of solving (\ref{p45})& Total computations \\
\hline
Linear search& $3$&$\rm {iter_1}$&$3 ({\rm N}+{\rm M} -1)+\rm iter_1 ({\rm N}+{\rm M}-1)$\\
\hline
Assignment& $3$&\textemdash&$3$\\
\hline
Addition &$({\rm N}+{\rm M}-1)({\rm N}+{\rm M} )$& $({\rm N}+{\rm M}+13) \ \rm{ iter_1}$ &$({\rm N}+{\rm M})({\rm N}+{\rm M}-1)+({\rm N}+{\rm M}+13) \ \rm{iter_1}$\\
\hline
Division & ${\rm N}+{\rm M}-1$ & $\rm{ iter}_1$ &$({\rm N}+{\rm M} -1)\ \rm{ iter_1}$ \\
\hline
Multiplication& $({\rm N}+{\rm M}-1)^2$ & $({\rm N}+{\rm M} +12) \ \rm{ iter_1}$ &$({\rm N}+{\rm M}-1)^2+({\rm N}+{\rm M}+12) \ \rm{ iter_1}$\\
\hline 
\end{tabular}}
\end{table*}
\begin{table*}
\centering
\caption {Computations of previous RNP schemes }\label{table3}

\resizebox{0.96\textwidth}{!}{
\begin{tabular}{|p{1.8cm}| p{3.5cm}|p{3.5cm}|p{3.5cm}|p{3.5cm}|}
\hline 
%
\multirow{2}{*}{Operation} & \multirow{2}{*}{Computations of DCA\cite{mohammadi2020increasing} }& \multicolumn{3}{c|}{Computations of RA \cite{liu2017optimal}}\\
\cline{3-5}
& & Computations of initializing the position of relay nodes on the surface of the water & Computations of adjusting the depth&Total computations \\
\hline
Linear search &$2 ({\rm N+M}-1)$&\textemdash&${\rm N+2M-2}$&${\rm N+2M-2}$\\
\hline 
Assignment&$4$&$2{\rm M}+4\alpha$&$\beta$&$2{\rm M}+4\alpha+\beta$\\
\hline
Addition & $({\rm N+M}) ({\rm N+M}-1)+({\rm N+M}+17) \ {{\rm iter_2}}$ & ${\rm N}({\rm 2N-4}+\alpha)+\alpha+\alpha^2$ & $2{\rm( N+M)(M+N-2)}+5\beta+{\rm 13 \ iter_3}$ & $2{\rm( N+M)(M+N-2)}+{\rm N}({\rm 2N-4}+\alpha)+\alpha+\alpha^2+5\beta+{\rm 13 \ iter_3}$\\
\hline
Division &$({\rm N+M}-1)+5 \ {\rm iter_2}$ &${\rm N+}3\alpha$&$\rm N+M$&${\rm2 N+M}+3\alpha$\\
\hline
Multiplication &$({\rm N+M}) ({\rm N+M}-1)+({\rm N+M}+23) \ {\rm iter_2}$&${\rm N}^2+4\alpha$&$({\rm N+M})^2+4\beta+14 \ \rm iter_3$&$({\rm N+M})^2+{\rm N}^2+4\alpha+4\beta+14 \ \rm iter_3$\\
\hline
\end{tabular}}
\end{table*}

\subsection{Step 2 - Relay nodes minimization: RNMI scheme}
\textcolor{black}{
By definition, $\rm M$ is an indication of the selected relay nodes which are in an acceptable network lifetime extension and we have 
$\rm M \leq M_0 $. Moreover, let $\mathbf{p}_c$ includes the total transmitting energy consumption of nodes $c_i, i=1, \ldots, \rm M_0$ before placing relays. Considering the effectiveness of relay nodes in the extension of network lifetime, a relay should be introduced if necessary in terms of the network lifetime. Toward this, we cast to the zero-norm and one-norm principles to define the relay node selection problem as: 
\begin{subequations}
\begin{align}
&\gamma=\arg \{\min \norm {\mathbf p_r}_0 , \max \norm {{\mathbf {p}_r-\mathbf {p}_c}}_1 \},\\
&s.t.\nonumber\\
&\min \tau_i=\tau^*, i\in \mathcal N,\\
&{[\mathbf {p}_{r}]_i}\in \{\mathbf 0, {[\mathbf{p}_{r}^*]_i}\},\\
& [\mathbf{p}_c]_i=
\begin{cases}
p_{c_{i}r_{i}} \times \sum_{j\in N_ {c_i}^u} [\mathbf{R}]_{{c_i}j} & [\mathbf {p}_{r}]_i=[\mathbf{p}_{r}^*]_i\\
\sum_{j \in N_{c_{i}}^{u}}^{}{p_{c_{i}j}}[\mathbf{R}]_{{c_i}j}& [\mathbf {p}_{r}]_i=0
\end{cases},i=1, \ldots, \rm M_0.
\end{align}
\end{subequations}
where zero-norm as the cardinality function returns the non-zero entry in the
$\mathbf p_r$. 
In addition, $\tau^*$ is the answer from the first optimization step added as the constraint. By employing the scalarization method to mix the zero-norm and one-norm functions, the problem can be given as:
\begin{subequations}
\begin{align}
&\gamma=\arg \{\min \omega_1 \norm {\mathbf {p}_r}_0 -\omega_2 \norm {{\mathbf {p}_r-\mathbf {p}_c}}_1 \},\\
&s.t.\nonumber\\
&\min \tau_i=\tau^*, i\in \mathcal N,\\
&{[\mathbf {p}_{r}]_i}\in \{\mathbf 0, {[\mathbf{p}_{r}^*]_i}\},\\
& [\mathbf{p}_c]_i=
\begin{cases}
p_{c_{i}r_{i}} \times \sum_{j\in N_ {c_i}^u} [\mathbf{R}]_{{c_i}j} & [\mathbf {p}_{r}]_i=[\mathbf{p}_{r}^*]_i\\
\sum_{j \in N_{c_{i}}^{u}}^{}{p_{c_{i}j}}[\mathbf{R}]_{{c_i}j}& [\mathbf {p}_{r}]_i=0
\end{cases},i=1, \ldots, \rm M_0.
\end{align}
\label{problem47}
\end{subequations}
where the weights for each function $\omega_1$ and $\omega_2$ can be chosen according to the kind of tradeoffs we are willing to make and $\omega_1+\omega_2=1$. 
Recall step function $s(x)$ with $s(x): \mathbb R \rightarrow \mathbb R^+$ that $s(x) =1$ for $x>0$ and $s(x)=0$ for $x\leq0$, the zero-norm can be written as the sum of discontinuous step function as:
\begin{equation}
\norm {\mathbf {p}_r}_0=\sum_{i=1}^{\rm M_0} s([\mathbf {p}_{r}]_i)
\end{equation}
Here, by applying the nonnegative feature of $[\mathbf {p}_{r}]_i$, we use the following continuously differentiable concave approximation of the step function for nonnegative variable \cite{stepapprox}:
\begin{equation}
s([\mathbf {p}_{r}]_i)=1-\exp (-\eta [\mathbf {p}_{r}]_i)
\end{equation}
that $\eta >0$. \textcolor{black}{ Therefore the problem (\ref{problem47}) is equivalently presented as: 
\begin{subequations}
\begin{align}
&\gamma=\arg\min \{ \omega_1\sum_{i=1}^{M_0} \left({1-\exp(-\eta [\mathbf {p}_{r}]_i)} \right)- \omega_2\norm {\mathbf {p}_r-\mathbf{p}_c}_1\},\\
&s.t.\nonumber\\
& \tau_i\geq \tau^*, i\in \mathcal N,\\
&{[\mathbf {p}_{r}]_i}\in \{\mathbf 0, {[\mathbf{p}_{r}^*]_i}\},\\
& [\mathbf{p}_c]_i=
\begin{cases}
p_{c_{i}r_{i}} \times \sum_{j\in N_ {c_i}^u} [\mathbf{R}]_{{c_i}j} & [\mathbf {p}_{r}]_i=[\mathbf{p}_{r}^*]_i\\
\sum_{j \in N_{c_{i}}^{u}}^{}{p_{c_{i}j}}[\mathbf{R}]_{{c_i}j}& [\mathbf {p}_{r}]_i=0
\end{cases},i=1, \ldots, \rm M_0.
\end{align}
\end{subequations}
The obtained model is a mixed-integer convex programming model. Models with integer and binary variables must still obey all of the same disciplined convex programming rules that CVX enforces for continuous models. 
The above approximation model is a smooth optimization problem with tolerable complexity and 
show that each relay node should not be deployed if 
the transmit power consumption by the cooperation of the relay node 
is more than the direct transmit power consumption and should be otherwise. }
Let us now evaluate our approach against 
the RA \cite{liu2017optimal} and DCA \cite{mohammadi2020increasing} schemes
in terms of complexity and network lifetime
}
\section{Complexity analysis}

In this section, we compute the time complexity of our proposed RNP, DCA and RA methods. We do this by calculating the number of computations as shown in Tables \ref{t_complexity} and \ref{table3}. One can notice that ${\rm iter}_i, i\in\{1, 2, 3\}$ represents the number of iterations required for each RNP scheme to converge, which is significantly lower than $({\rm N}+{\rm M})$ in practical scenarios. Additionally, $\alpha$ and $\beta$ are small constant values. Consequently, Table \ref{t_complexity} demonstrates that the worst-case time complexity of our proposed method is limited to ${\mathcal O} ({\rm N}+{\rm M})^2$. Furthermore, based on table \ref{table3}, it can be concluded that the time complexity of DCA and RA is also bounded by ${\mathcal O} ({\rm N+M}) ^2 $. Subsequently, we will assess the performance of the proposed ORNS method through simulations.
\begin{table}[!t]
\caption{List of parameters} 
\centering 
\resizebox{0.95\columnwidth}{!}{
\begin{tabular}{|c | c| c|} 
\hline
Notation &Definition & Value \\
\hline 
$H_s$ &Depth of the water &2000 m \\
\hline
$C_R$ &Communication range &500 m \\
\hline
$f$ &Frequency of acoustic signal & 1 kHz \\
\hline
$p_s$ &The power consumption for processing in sending data &1 mw/bit\\
\hline
$p_r$ &The power consumption for processing in receiving data &1 mw/bit \\
\hline
$d_t$ & Distance threshold & 87 m \\
\hline
$L_c$ & Link capacity & 10 kb/sec\\
\hline
$\epsilon_p$ & Primary energy of typical node $i$ & $4\times10^5$ J \\
\hline
${g}_i$ & Generation rate of sensor node $i$ &$10\thicksim200$ bit/sec \\
\hline 
\end{tabular}}
\label{t2}
\end{table}
\section{Algorithm evaluation results }
In this section, we evaluate and compare the proposed RNP with the heuristic RA method \cite{liu2017optimal} and DCA approach \cite{mohammadi2020increasing} using multiple simulation scenarios. 
The experiments are performed by MatLab 2017b. 
The simulation parameters and their notations are provided in Table \ref{t2}. The depth of water was taken as $2000 \ \rm m$, the generation rate of each sensor node is set randomly between $10$ and $200 \ \rm bit/Sec$, and the primary energy of nodes was set to $4\times 10^5 \rm J$ \cite{alsalih2010placement}. The frequency of the acoustic signal was set to $1 \rm kHz$. 
Similar to \cite{liu2017optimal}, \cite{mohammadi2020increasing}, we design the deployment of 3-D underwater sensor networks in the cylindrical sensing field where they sense the environment. 
The gathered data is then transmitted to the SB, which is positioned at the origin.
Additionally, Table \ref{T3} outlines the different RNP cases, using $\gamma_r=\frac{|\mathcal R|}{|\mathcal S|}$ to denote the percentage of employed relays and $\rm{RF} = \frac{\epsilon_{c_1}}{\epsilon_{p}},$ as the residual energy factor of the most critical node in the sensor network. 

The Imbalanced Factor of Energy Consumption (IEC) was calculated as
\begin{equation}
\mathrm{IEC}=\frac{\frac{1}{\rm N}\sum_{i\in \mathcal S}{\left(\mathrm{E}\left(\epsilon_i\right)-\mathrm{E}\left(\bar \epsilon\right)\right)^2}}{\sigma_0^2}
\end{equation}
where
$
\bar \epsilon=\frac{1}{{\rm N}}\sum_{i\in \mathcal S}{\epsilon_i},
$
and $\sigma_0^2$ is the normalization factor.
In the following we present a three-fold approach to evaluate the effectiveness of our relay node placement method. Firstly, we investigate regulation of the positions of relay nodes within the network at different scales. 
Secondly, we conduct a comprehensive performance evaluation of our proposed method by comparing it with existing approaches.
Thirdly, we delve into the details of our relay node selection design, which aims to minimize resource wastage in UASNs. 

\begin{table}[!t]
\centering
\caption {Considered RNP characteristics}\label{t1}
\resizebox{0.9\columnwidth}{!}{
\begin{tabular}{p{1.2cm} p{1.2cm} p{1.2cm} p{1.2cm} p{1.2cm}}
\hline 
\hline 
\centering
Case& A& B & C & D \\
\hline
\centering
$\rm{RF}$& 0.25 &0.75& 0.25 & 0.25 \\
\centering
$\gamma_r$ & 0.3 &0.3& 0.6 & 0.9 \\
\hline
\centering
\end{tabular}}
\label{T3}
\end{table}
\begin{figure*}[!t]
\center
\begin{subfigure}{0.32\textwidth}
\centering{\includegraphics[width = 0.99\columnwidth]{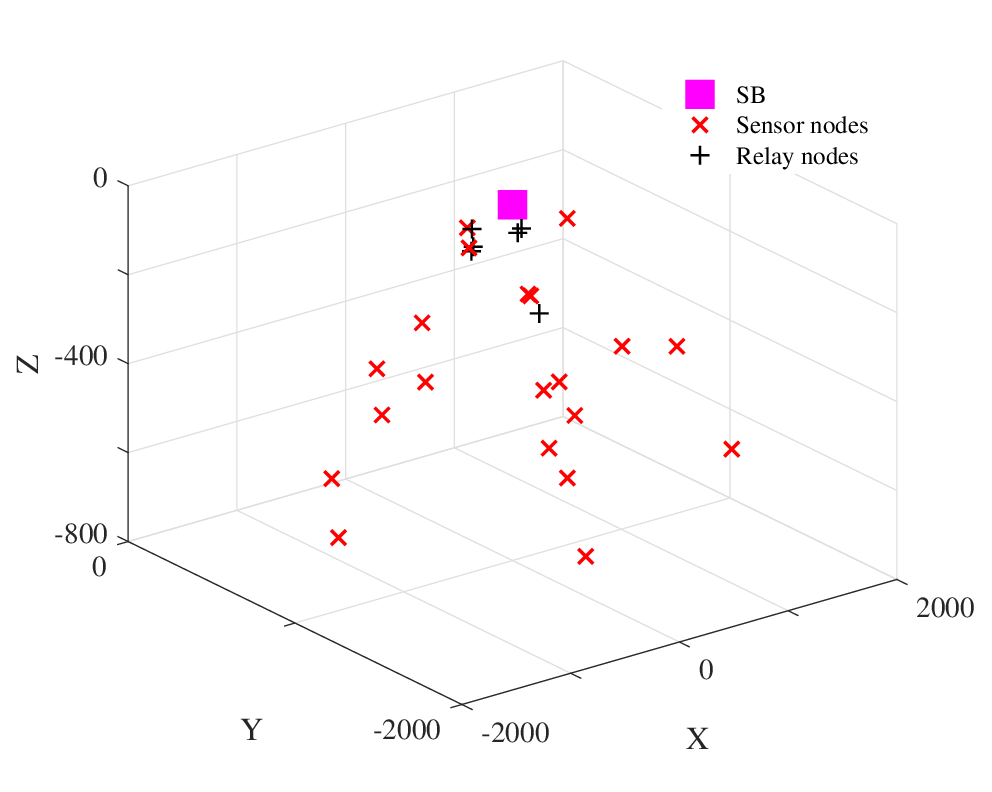}}
\vspace{-0.5\baselineskip}
\caption{RA method (first deployment)}
\label{effect_radius_lifetime}
\end{subfigure}
\begin{subfigure}{0.32\textwidth}
\centering{\includegraphics[width = 0.99\columnwidth]{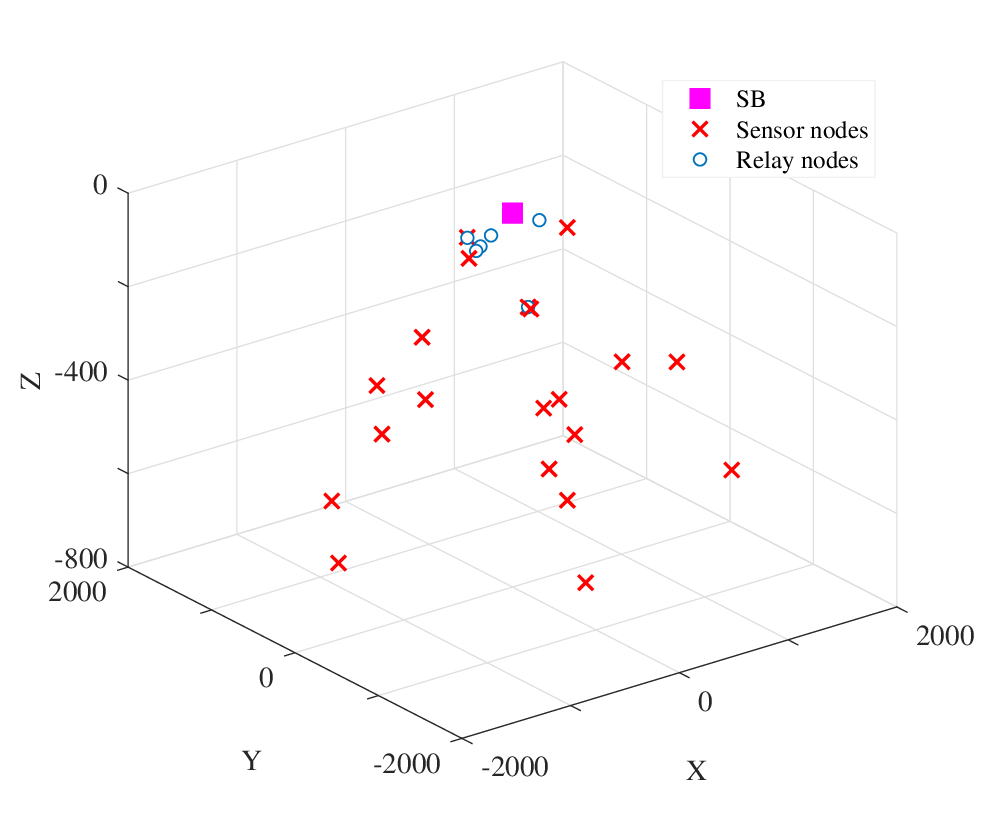}}
\vspace{-0.5\baselineskip}
\caption{DCA approach (first deployment)}
\label{effect_radius_residual}
\end{subfigure}
\begin{subfigure}{0.32\textwidth}
\centering{\includegraphics[width = 0.99\columnwidth]{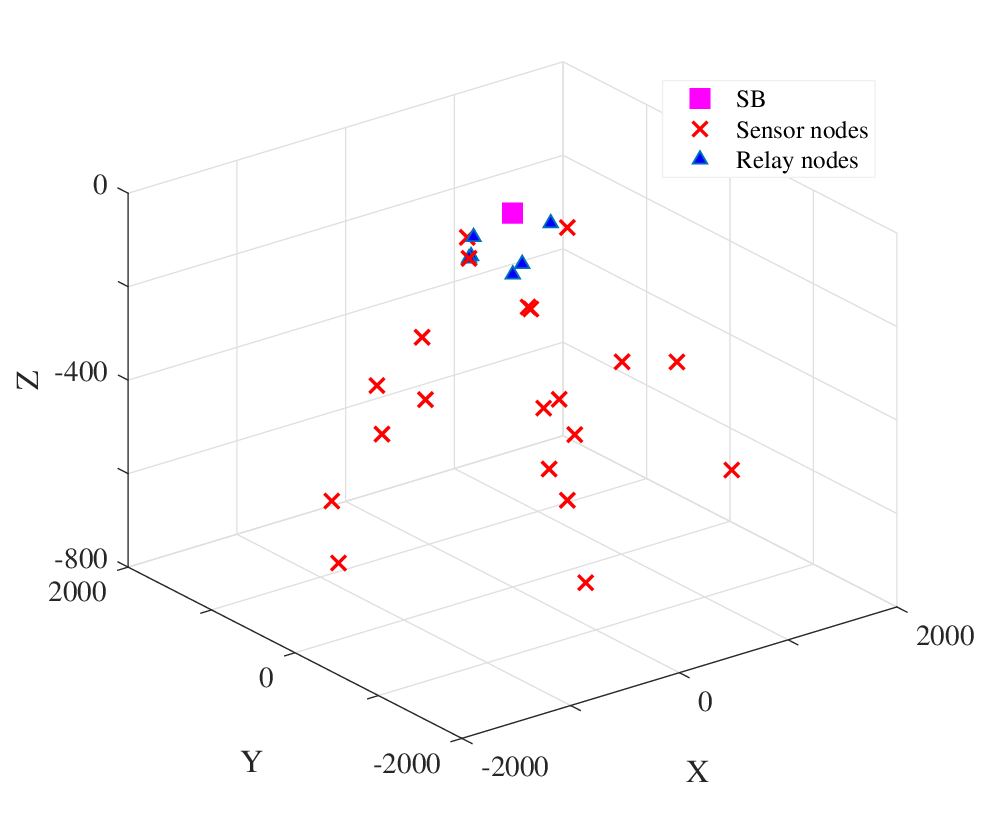}}
\vspace{-0.5\baselineskip}
\caption{Proposed approach (first deployment)}
\label{effect_radius_lifetime}
\end{subfigure}
\begin{subfigure}{0.32\textwidth}
\centering{\includegraphics[width = 0.99\columnwidth]{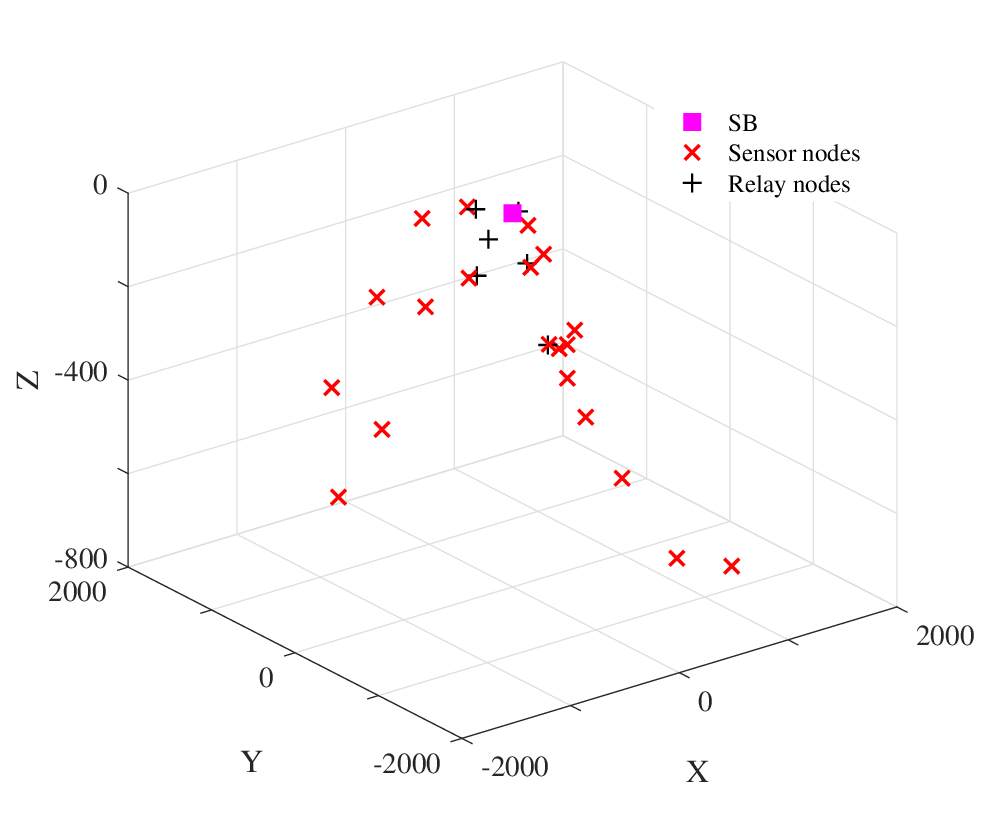}}
\vspace{-0.5\baselineskip}
\caption{RA method (second deployment)}
\label{Fig4-revised-ra}
\end{subfigure}
\begin{subfigure}{0.32\textwidth}
\centering{\includegraphics[width = 0.99\columnwidth]{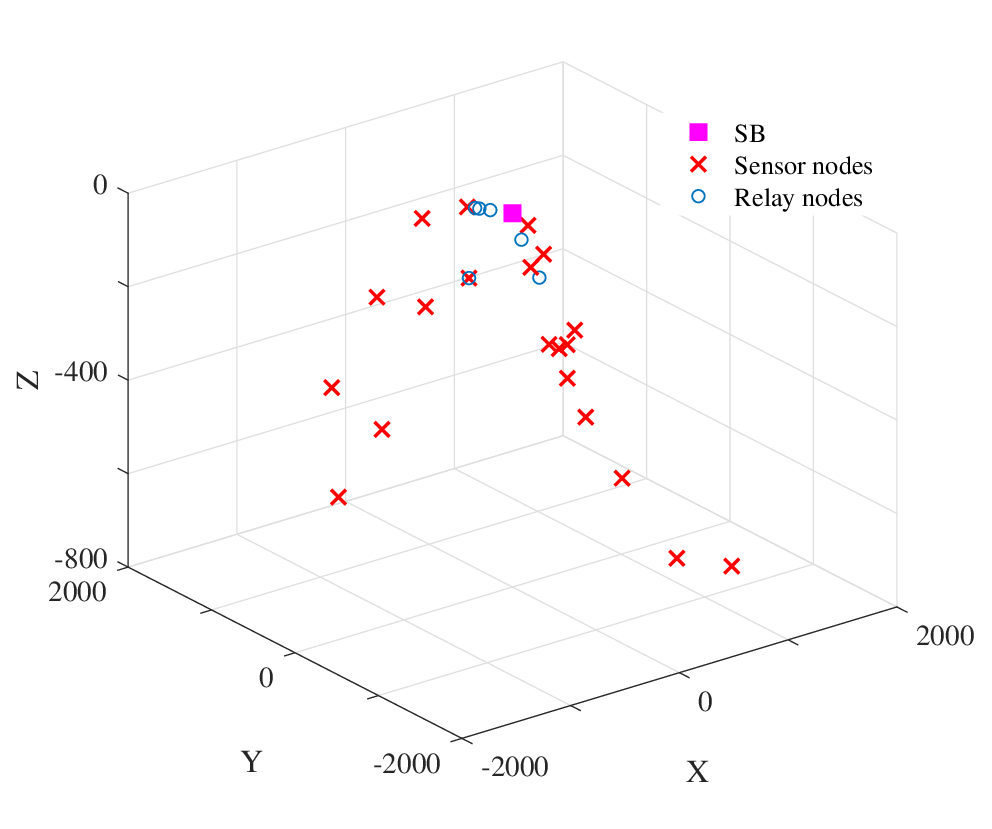}}
\vspace{-0.5\baselineskip}
\caption{DCA approach (second deployment)}
\label{Fig4-revised-dca}
\end{subfigure}
\begin{subfigure}{0.32\textwidth}
\centering{\includegraphics[width = 0.99\columnwidth]{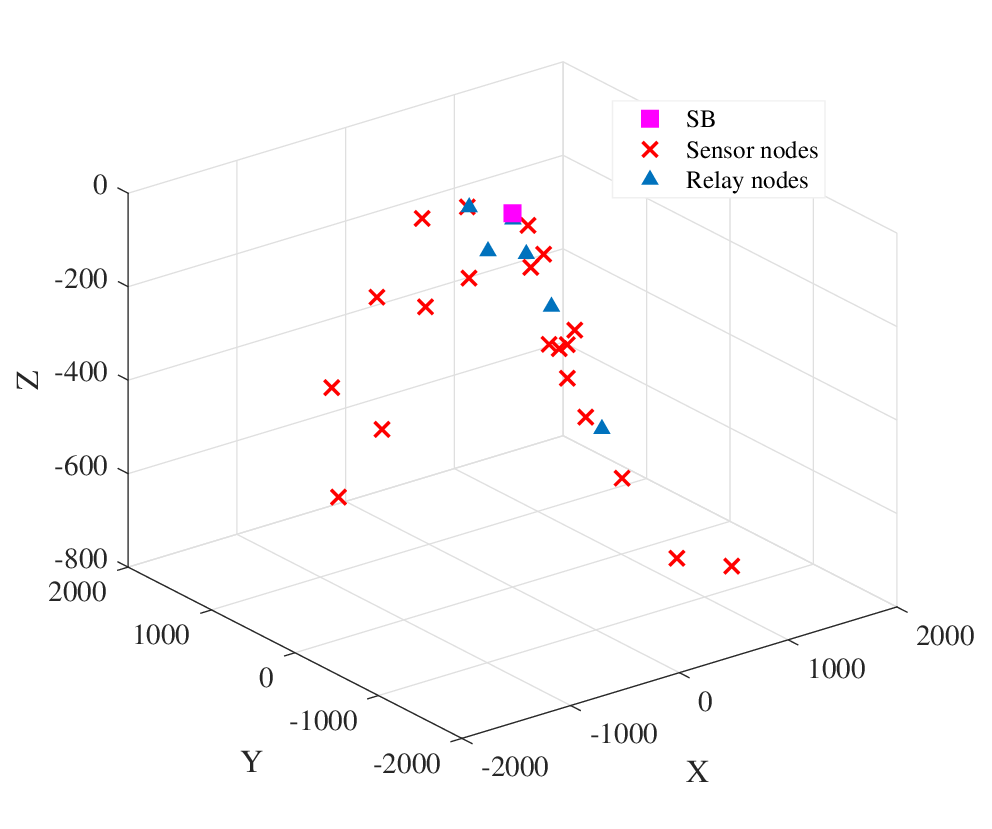}}
\vspace{-0.5\baselineskip}
\caption{Proposed approach (second deployment)}
\label{Fig4-revised-pr}
\end{subfigure}
\caption{Position of relay nodes in the network when $|\mathcal S|=20$, $\gamma_r=0.3$, and $\rm{RF}=0.25$ in different deployments}
\label{Fig4-revised}
\end{figure*}
\begin{figure*}[!t]
\center
\begin{subfigure}{0.32\textwidth}
\centering{\includegraphics[width = 0.99\columnwidth]{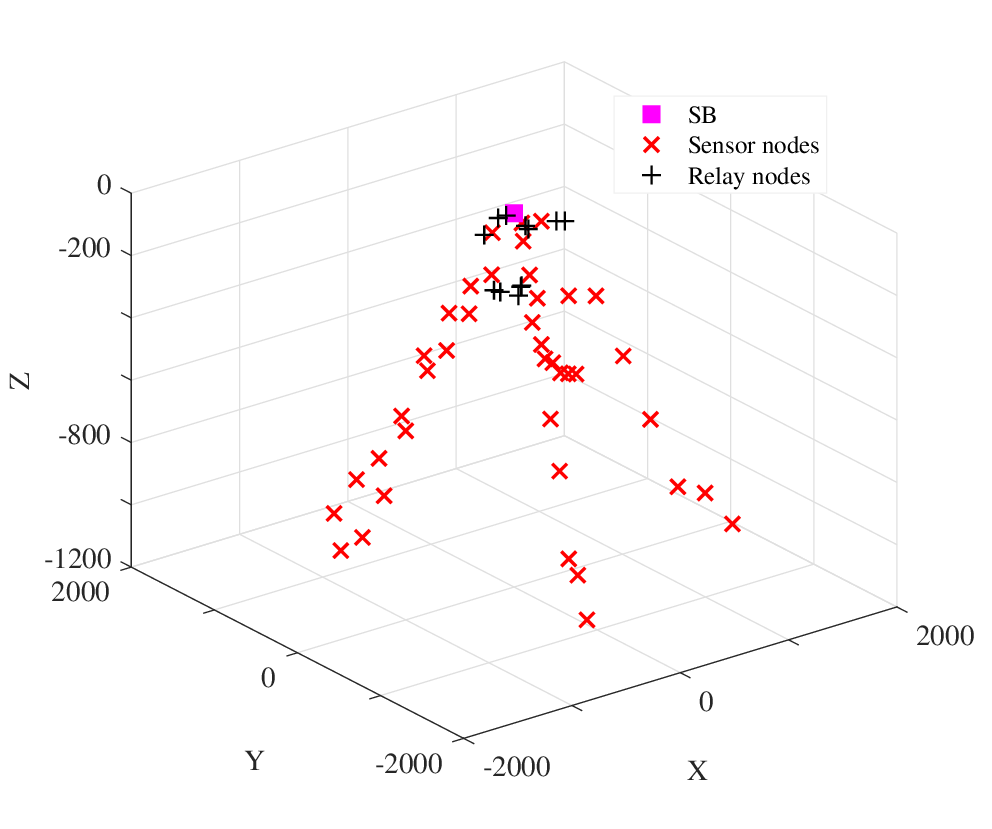}}
\vspace{-0.5\baselineskip}
\caption{RA method (first deployment)}
\label{Fig5a-revised}
\end{subfigure}
\begin{subfigure}{0.32\textwidth}
\centering{\includegraphics[width = 0.99\columnwidth]{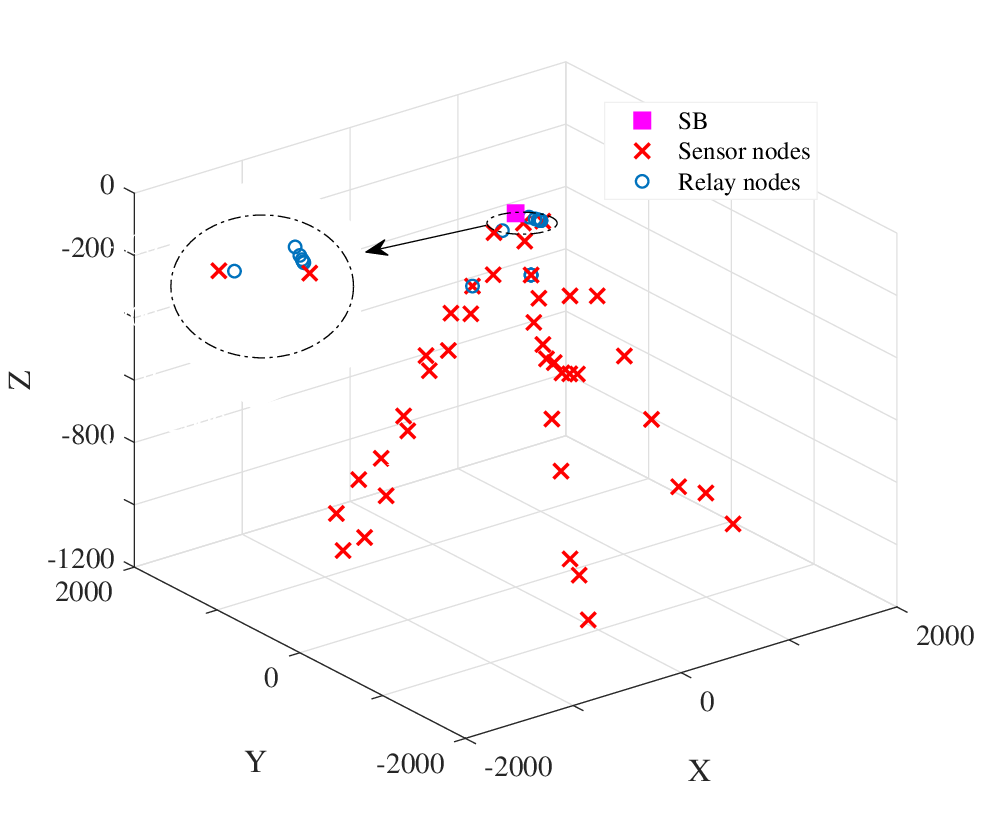}}
\vspace{-0.5\baselineskip}
\caption{DCA approach (first deployment)}
\label{Fig5-revised-dca}
\end{subfigure}
\begin{subfigure}{0.32\textwidth}
\centering{\includegraphics[width = 0.99\columnwidth]{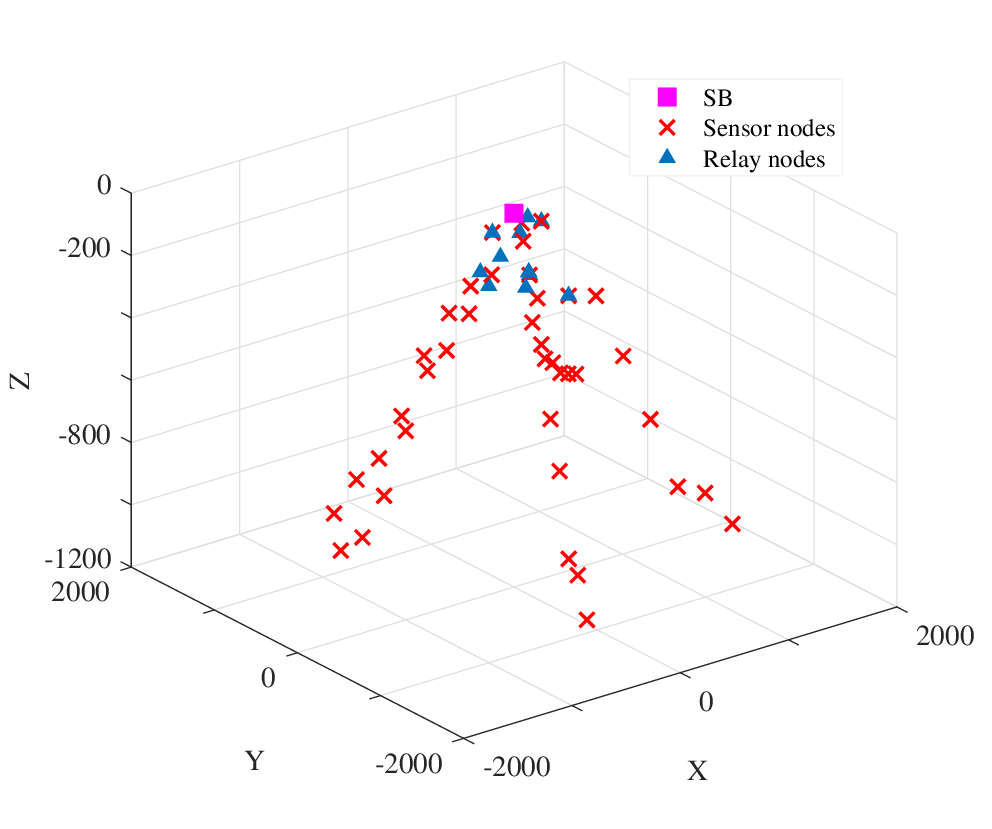}}
\vspace{-0.5\baselineskip}
\caption{Proposed approach (first deployment)}
\label{Fig5-revised-pr}
\end{subfigure}
\begin{subfigure}{0.32\textwidth}
\centering{\includegraphics[width = 0.99\columnwidth]{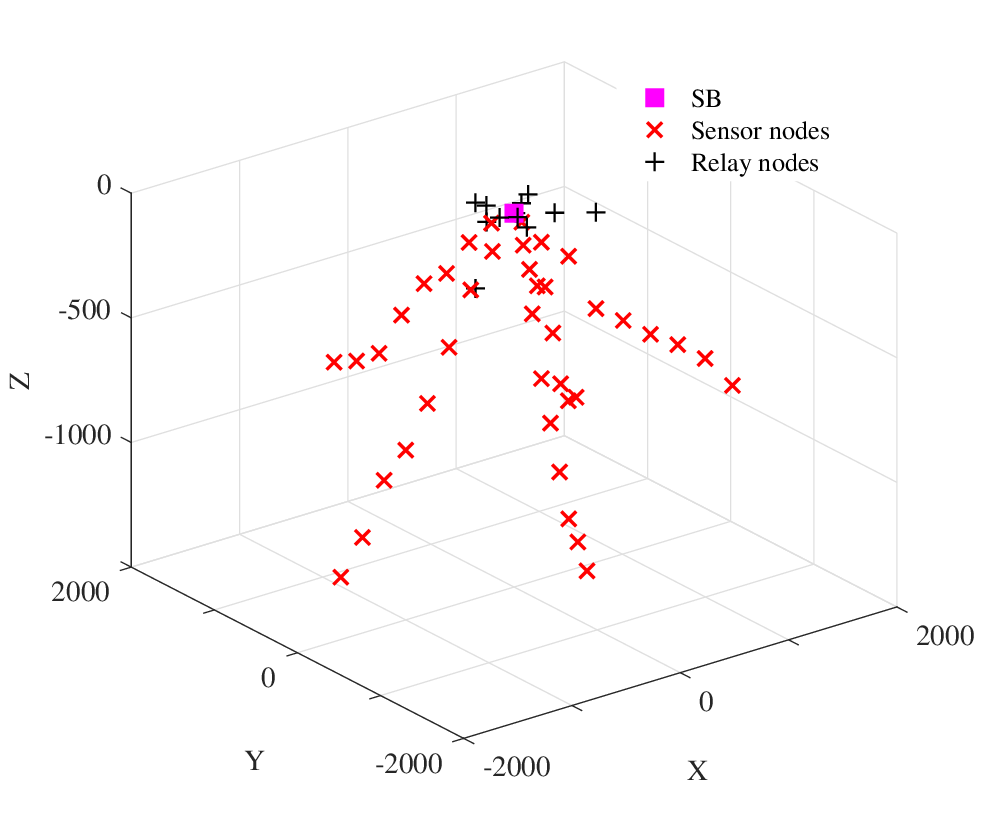}}
\vspace{-0.5\baselineskip}
\caption{RA method (second deployment)}
\label{Fig5d-revised}
\end{subfigure}
\begin{subfigure}{0.32\textwidth}
\centering{\includegraphics[width = 0.99\columnwidth]{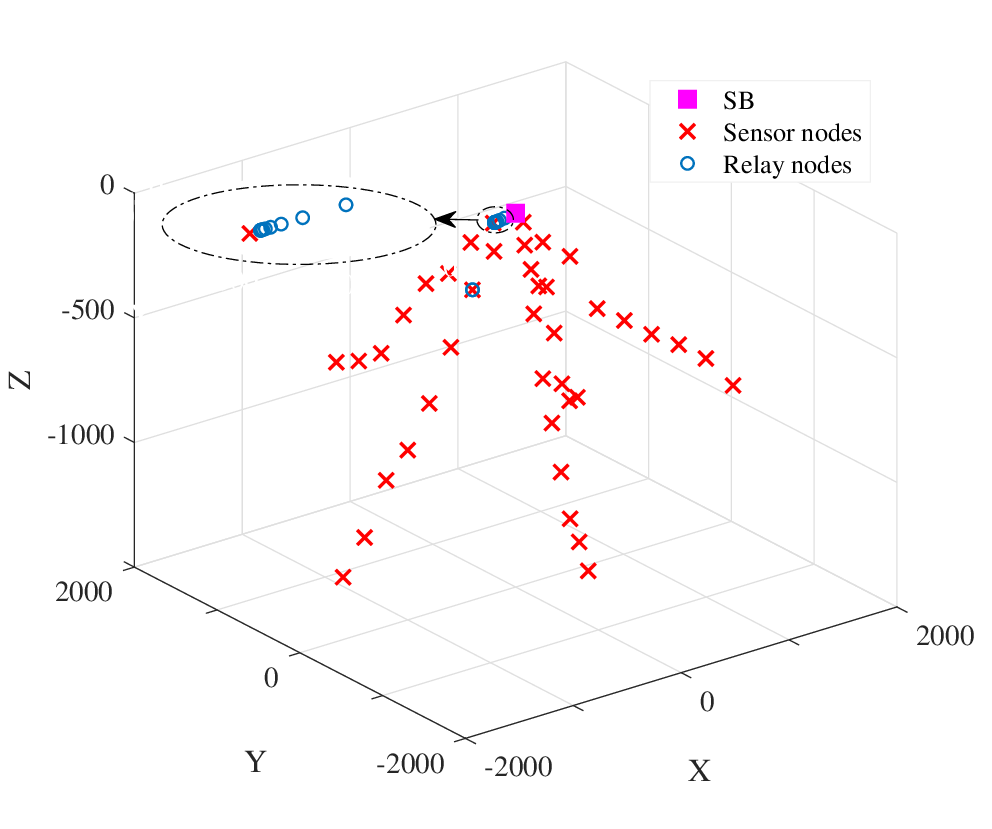}}
\vspace{-0.5\baselineskip}
\caption{DCA approach (second deployment)}
\end{subfigure}
\begin{subfigure}{0.32\textwidth}
\centering{\includegraphics[width = 0.99\columnwidth]{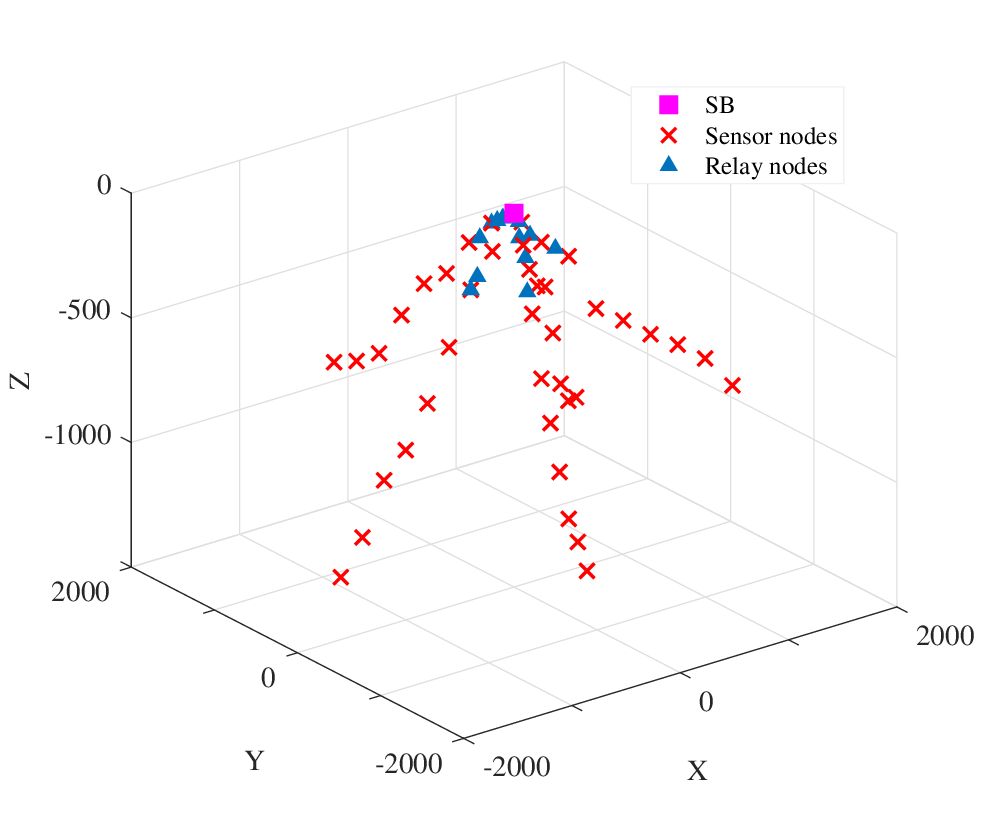}}
\vspace{-0.5\baselineskip}
\caption{Proposed approach (second deployment)}
\end{subfigure}
\caption{Position of relay nodes in the network when $|\mathcal S|=40$, $\gamma_r=0.3$, and $\rm{RF}=0.25$ in different deployments}
\label{Fig5-revised}
\end{figure*}

\begin{figure}[!t]
\center
\begin{subfigure}{0.45\textwidth}
\centering{\includegraphics[width = 0.99\columnwidth]{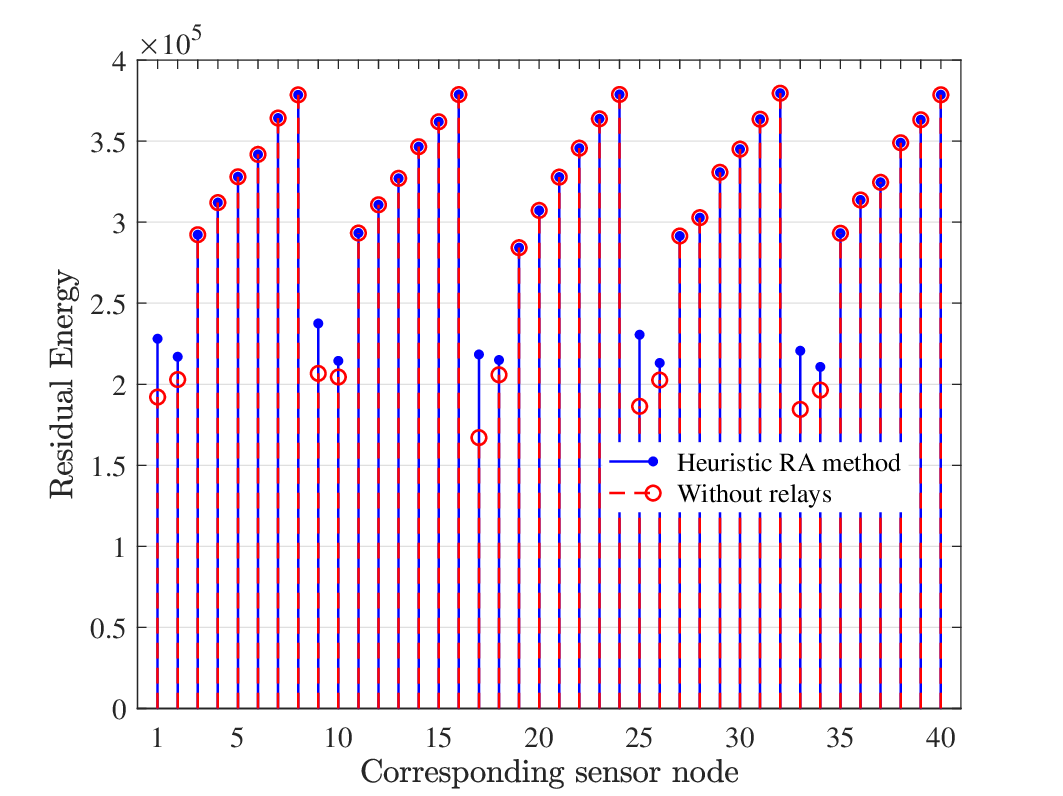}}
\vspace{-0.5\baselineskip}
\caption{RA method}
\label{Fig2-simulationA}
\end{subfigure}
\begin{subfigure}{0.45\textwidth}
\centering{\includegraphics[width = 0.99\columnwidth]{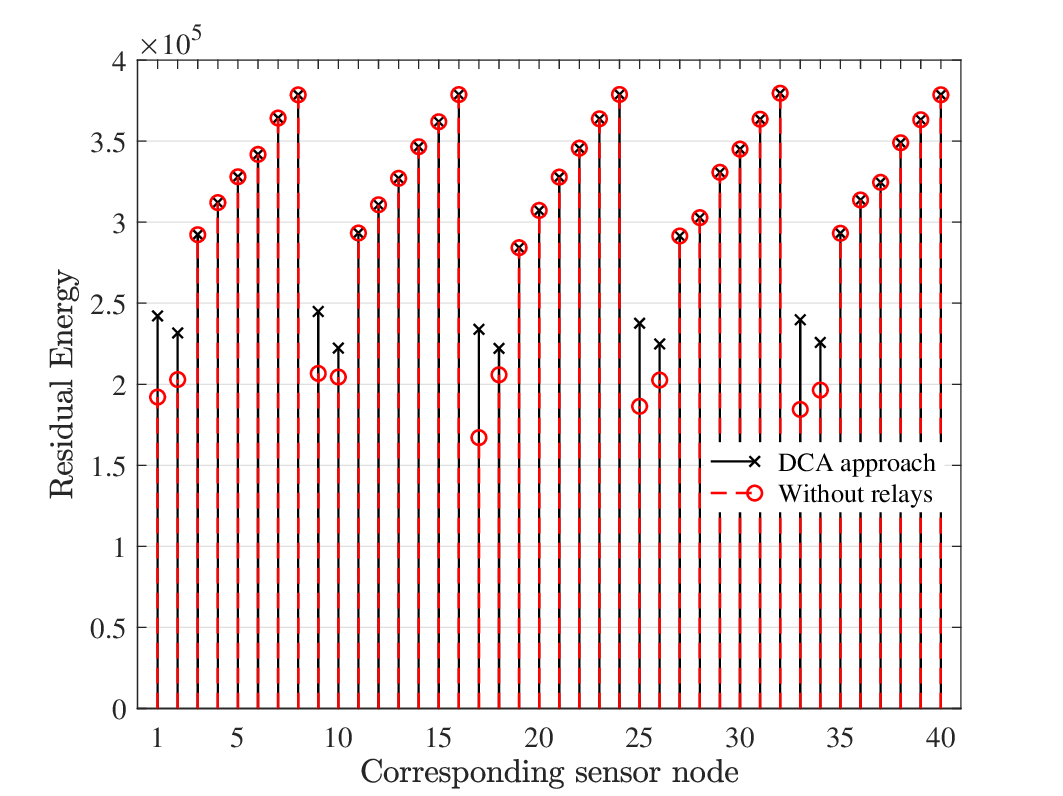}}
\vspace{-0.5\baselineskip}
\caption{DCA approach}
\label{Fig2-simulationB}
\end{subfigure}
\begin{subfigure}{0.45\textwidth}
\centering{\includegraphics[width = 0.99\columnwidth]{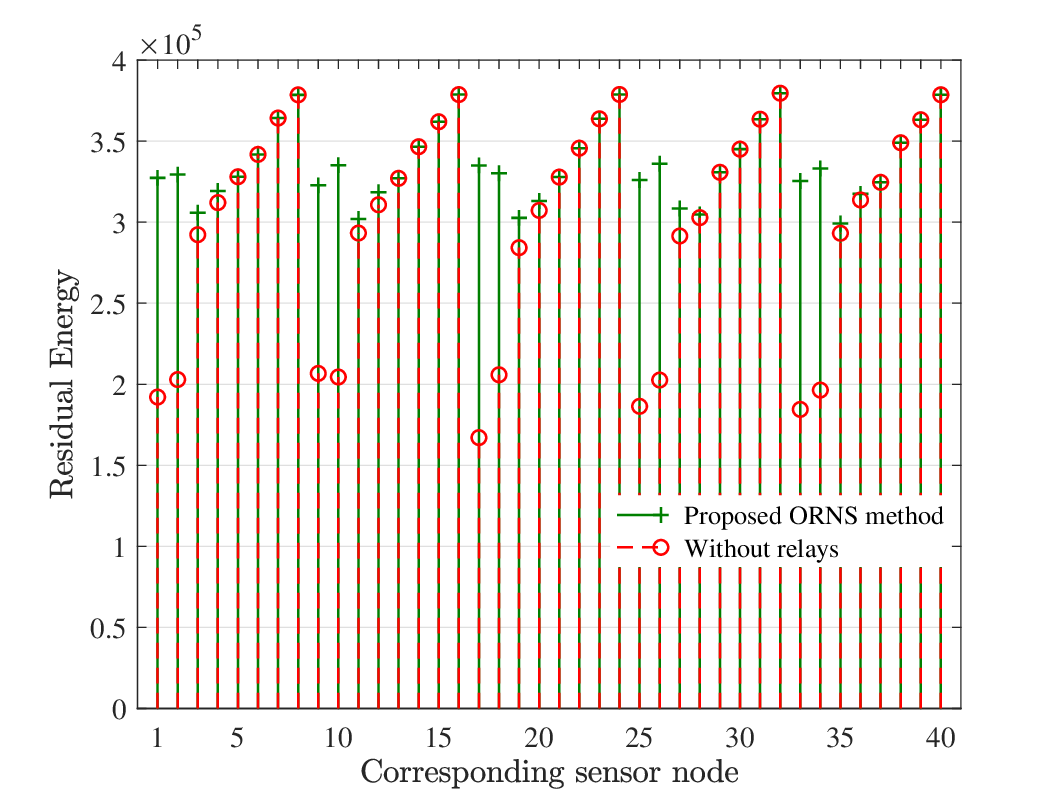}}
\vspace{-0.5\baselineskip}
\caption{Proposed approach}
\label{Fig2-simulationC}
\end{subfigure}
\caption{Residual energy of sensor nodes when $|\mathcal S|=40$, $\gamma_r=0.3$, and $\rm{RF}=0.25$}
\label{Fig2-simulation}
\end{figure}
\begin{figure}[!t]
\centering{\includegraphics[width=0.9\columnwidth]{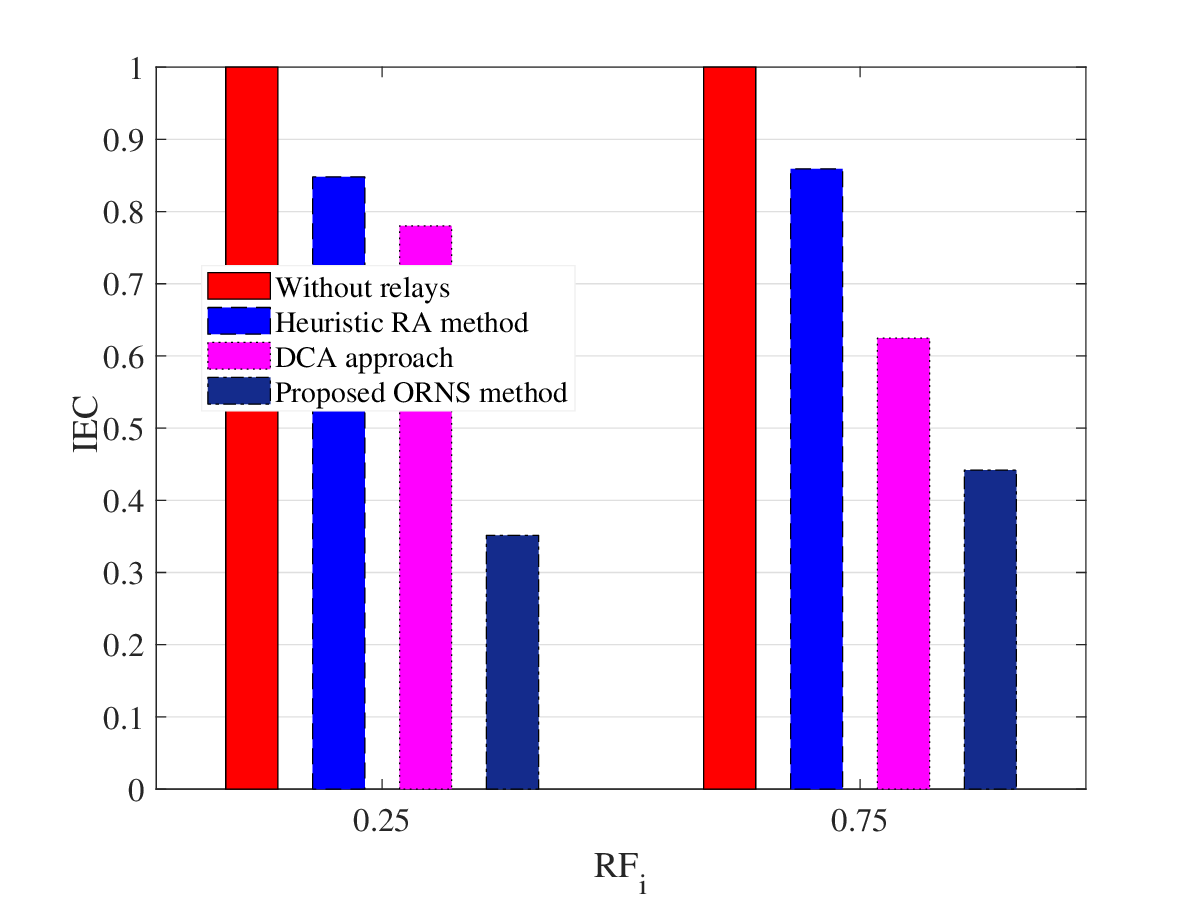}}
\caption{Imbalanced factor of energy consumption for different $\mathrm{RF}$s}
\label{Fig3-simulation}
\end{figure}
\begin{figure}[!t]
\begin{subfigure}{0.47\textwidth}
\centering{\includegraphics[width = 0.9\columnwidth]{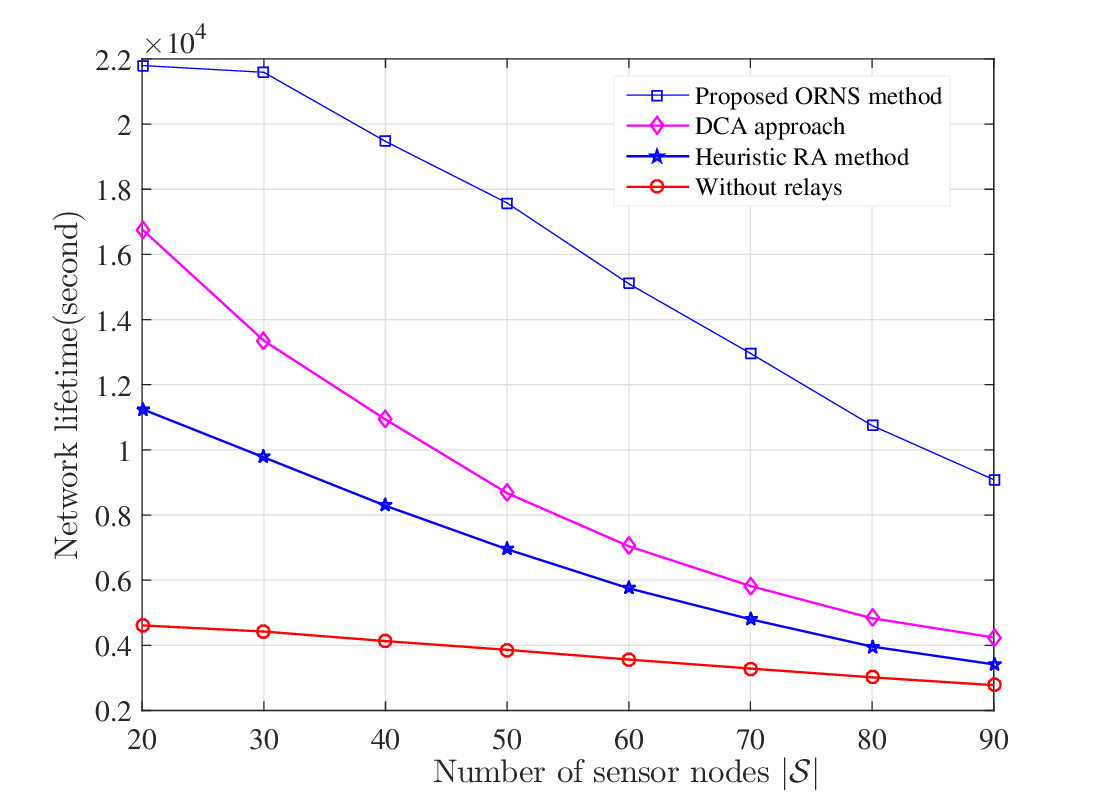}}
\caption{$\gamma_r=0.3, \mathrm{RF}=0.25$}
\label{Fig4-simulationA}
\end{subfigure}
\begin{subfigure}{0.47\textwidth}
\centering{\includegraphics[width = 0.9\columnwidth]{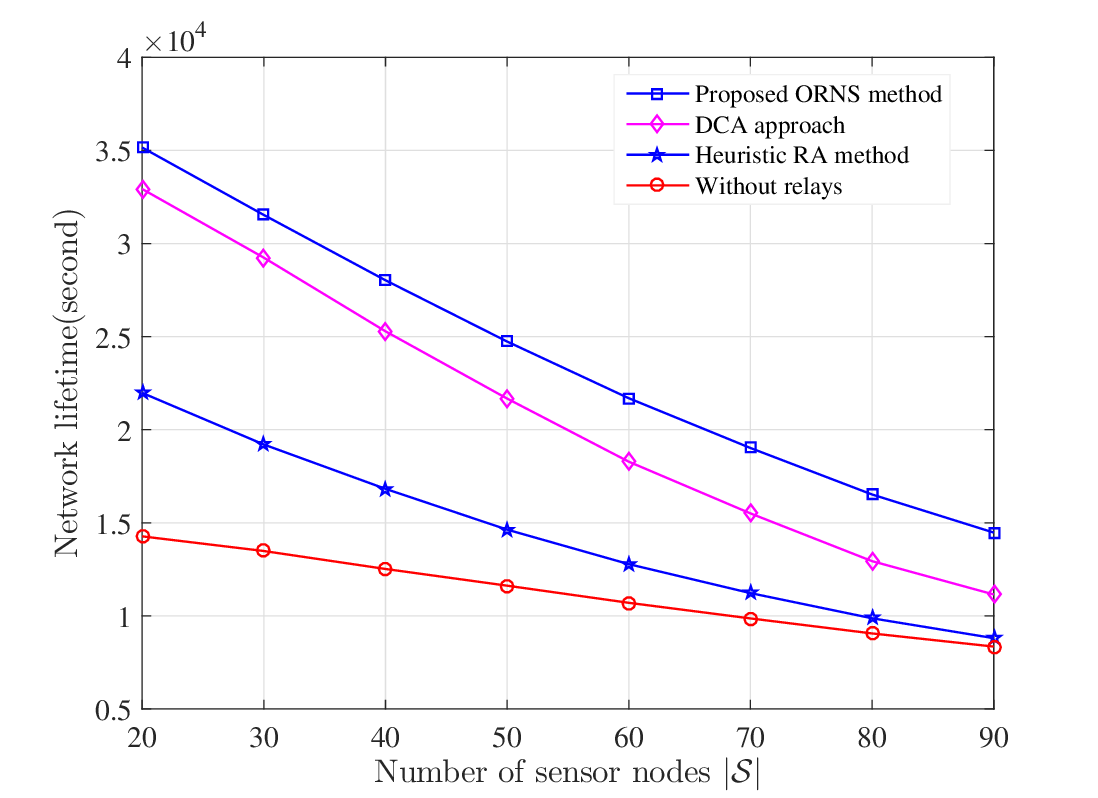}}
\caption{$\gamma_r=0.3, \mathrm{RF}=0.75$}
\label{Fig4-simulationB}
\end{subfigure}
\caption{Investigation of the performance in different network scales}
\label{Fig4-simulation}
\end{figure}

\subsection{Regulation of relay nodes}
In this section, we examine the regulation of relay nodes in the network. To do so, we consider Case A with $\rm RF=0.25$ and $\gamma_r=0.3$. We consider two network scales in this example, each with a different number of sensor nodes. For each scale, we present the results for two network deployments, as shown in Figs. \ref{Fig4-revised} and \ref{Fig5-revised}. In Fig. \ref{Fig4-revised}, we assume that there are 20 sensor nodes, while in Fig. \ref{Fig5-revised}, we assume that there are 40 sensor nodes.
\\
We observe that in a multi-hop UASN, a higher percentage of relay nodes are positioned near the SB due to the larger amount of data collected by nodes in that area. Comparing Fig. \ref{Fig4-revised} with Fig. \ref{Fig5-revised}, as the number of sensor nodes increases, the percentage of relay nodes near the SB also increases. This observation can be attributed to the fact that an increase in the number of sensors leads to a larger amount of information being gathered by nodes near the SB.
In the previous RA method, most relay nodes are not adjusted in depth because of their random location on the water's surface. Furthermore, previous line-segment relay node placement (LSRNP) approaches, such as RA and DCA, resulted in relay nodes being positioned too closely together or even directly on top of sensor nodes, as seen in Fig. \ref{Fig5a-revised} and Fig. \ref{Fig5d-revised}. In contrast, our proposed approach suggests a more sensible placement for relay nodes by considering a feasible search for the convex hull.

\subsection {Performance evaluation}

\begin{figure}[htbp]
\centering{\includegraphics[width = 0.95\columnwidth]{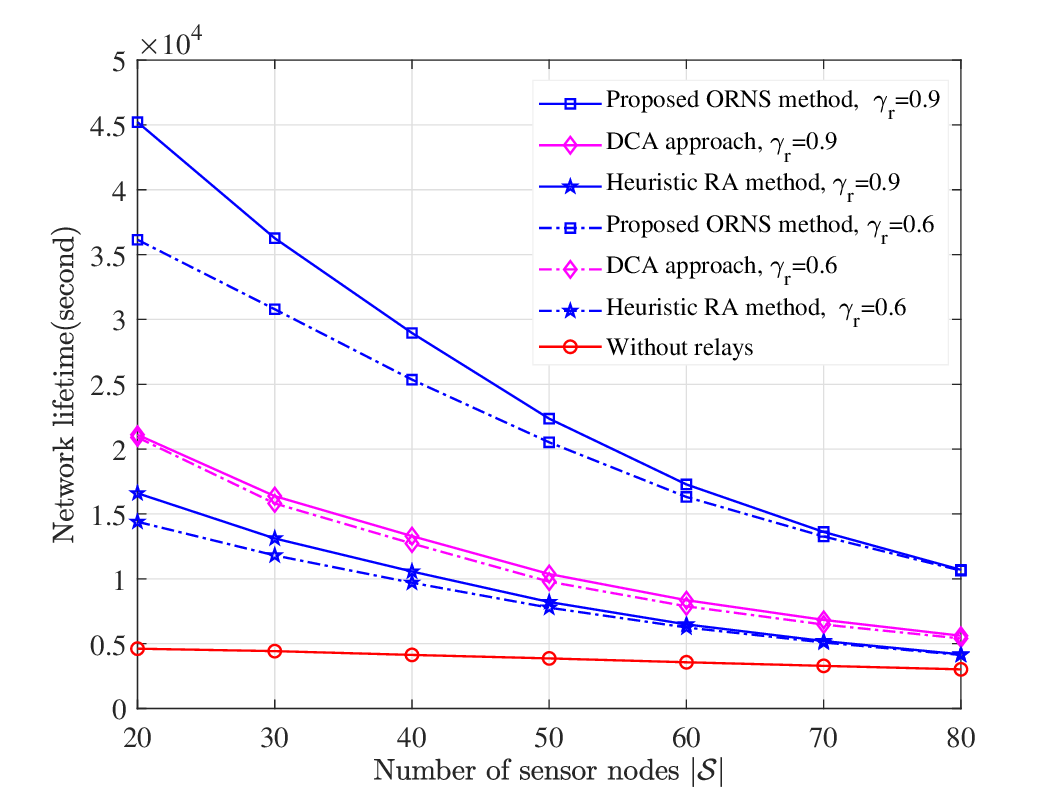}}
\caption{Effect of increasing the number of relay nodes on the performance}
\label{Fig5-simulation}
\end{figure}

In this part, we evaluate the performance of proposed RNP method and compare with RA and DCA approaches in terms of 
balanced energy consumption and network lifetime. 
To do so, we design some simulation cases and obtain all the numerical results from 50 deployments.
The first case assumes 40 sensor nodes, $\gamma_r=0.3$, and $\mathrm{RF}=0.25$, as shown in Case A in Table \ref{T3}. Figs. \ref{Fig2-simulationA}, \ref{Fig2-simulationB} and \ref{Fig2-simulationC} show the distribution of residual energy among sensor nodes when applying the RA, DCA, and ORNS methods, respectively, along with a scenario without relay nodes for comparison. It is observed that imbalanced energy consumption creates an energy hole in the network that limits the performance of UASNs. This, in turn, prevents the collected data from being forwarded to the SB. Employing relay nodes improves the energy hole issue in the network. Furthermore, the proposed ORNS method highlights the advantages of obtaining a more balanced energy distribution among nodes through optimal positioning of the relay nodes in addressing the energy hole issue. To illustrate further, Fig. \ref{Fig3-simulation} depicts the IEC factor for different RF values. It can be seen that the proposed method outperforms previous schemes in terms of energy balance, as evidenced by the smaller IEC values.
\par
We present the network lifetime as a function of the number of sensor nodes in the second case, with the parameters set according to cases A and B in Table \ref{T3}. The results for these cases are shown in Fig. \ref{Fig4-simulationA} and \ref{Fig4-simulationB}, respectively. Several important observations can be made based on these results. Firstly, it can be seen that the network lifetime decreases as the number of sensor nodes increases. This is due to the increased amount of data that needs to be relayed, leading to higher energy consumption. Secondly, the proposed RNP method outperforms RA and DCA by utilizing convex programming to determine the optimal location of the relay nodes. Lastly, without relay nodes, the transmission distance between the sensor nodes becomes long, resulting in a shorter network lifetime.
\subsection{Relay node selection design}
So far, we have assumed a fixed number of relay nodes in our scenario. However, in this section, we want to highlight the importance of selecting relay nodes in our proposed approach. To address this concern, we conducted a thorough analysis by comparing the network lifetime with varying numbers of sensor nodes and relay nodes.
The results, shown in Fig.\ref{Fig5-simulation} for cases C and D (where $\gamma_r=0.6$ and $\gamma_r=0.9$), clearly demonstrate that when the network has a high number of nodes, especially when the communication distance between sensor nodes is small, the need for additional relay nodes decreases. This is because the close proximity of sensor nodes allows for direct communication without relying intermediate relay nodes. In such cases, the presence or absence of extra relay nodes does not significantly affect the network lifetime.
\par
To gain further insight, consider Fig. \ref{Fig6-simulation} where we assume there are 80 sensor nodes and plot the positions of relay nodes in the network. It can be observed that the relay nodes are placed very closely together (DCA) or have specific positions beyond a certain number (proposed approach). In this situation, our relay node selection design, which strategically places five relay nodes, proves effective in achieving optimal network performance while minimizing the overall number of required relay nodes.
On the other hand, when relay node selection is not considered, there are no criteria to limit the placement of relay nodes. In conclusion, by carefully selecting these relay nodes based on our proposed criteria, we were able to achieve significant improvements in both maximizing network lifetime and minimizing the number of required relay nodes.

\begin{figure}[!t]
\centering{\includegraphics[scale=.47]{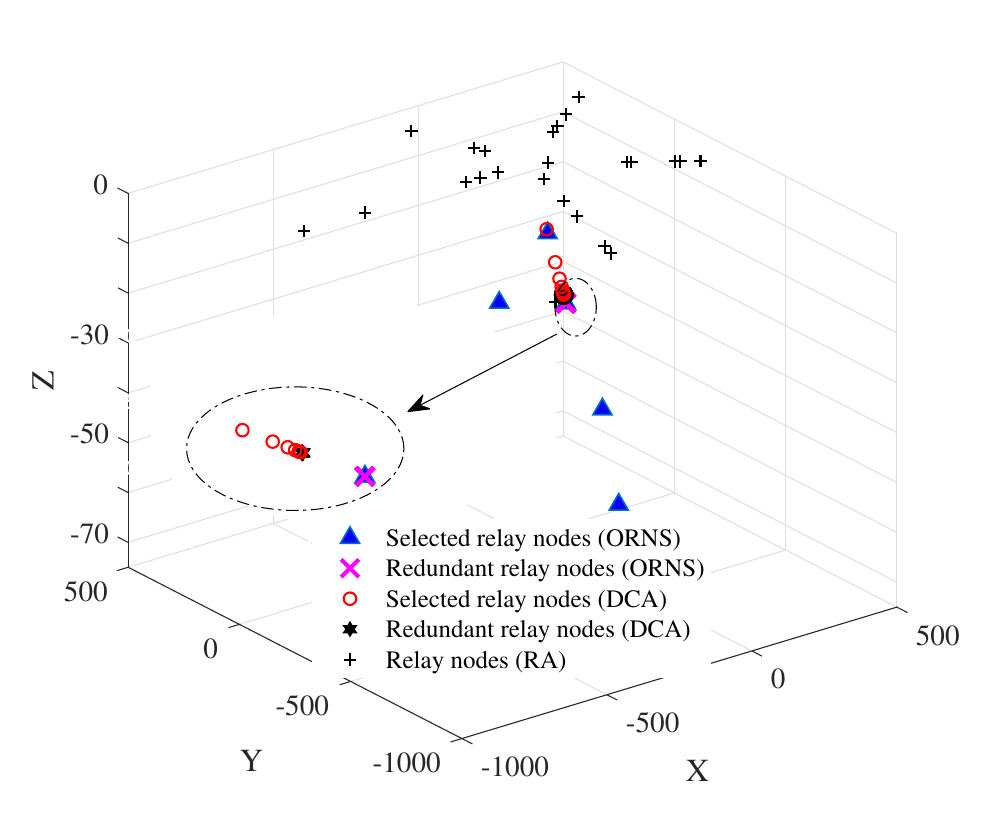}}
\caption{Distribution of relay nodes when $|\mathcal S|=80$ and $\mathrm{RF}=0.25$}
\label{Fig6-simulation}
\end{figure}
\section{Conclusion}
This paper aimed to address the joint optimization of maximizing network lifetime and minimizing relay node deployment in relay-assisted UASNs. To achieve a Pareto optimal solution, a multi-objective lexicographic method was employed in which the primary goal was to optimize the network lifetime in the RNP followed by the reduction of the active relay nodes. To accomplish this, a two-step process was employed. First, the ORNS algorithm was utilized to formulate the position of each relay node as a non-convex programming problem. This was then converted into an equivalent convex programming problem using a novel transformation and epigraph form scheme. Subsequently, a relay selection procedure utilizing a mixed-integer convex programming model was applied to minimize the number of active relay nodes. The proposed approach proved to be more efficient in terms of network lifetime than the existing models (RA and DCA).
\appendices
\section{Equivalent convex form of the (\ref{p32})}\label{A1}
Here, we present a transformation of a non-convex optimization problem with a convex objective function and convex (and/or linear) and DC constraints into an equivalent convex optimization problem. The non-convex problem can be expressed as 
\begin{subequations}
\begin{align}
& \mathbf{x}={\rm arg\min} f_0( \mathbf x)\\\intertext{s.t.} 
&f_i(\mathbf x)\leq0,\quad \quad i=1,2, \ldots, I,\\
&h_p(\mathbf x)=0,\quad \quad p=1,2, \ldots, P,\label{DC1}
\end{align}
\label{LCDC}
\end{subequations}
where $\mathbf x\in \mathbb{R}^n$ , $f_0, ..., f_I: \mathbb{R}^n\rightarrow\mathbb{R}$ 
are convex functions and $h_1, \ldots, h_M: \mathbb{R}^n\rightarrow\mathbb{R}$ are DC functions. We take advantage of the property of DC constraints which can be expressed as the difference of convex functions, and introduce a new variable $\bold{t}=[t_1,t_2, \ldots ,t_M]^T \in \mathbb{R}^M$. This yields a convex optimization problem, 
\begin{subequations}\label{p49}
\begin{align}
&( \mathbf x,\bold {t})=\rm arg\min f_0( \mathbf x)\\\intertext{s.t.} 
&f_i( \mathbf x)\leq0,\quad \quad i=1,2, \ldots, I,\\
&\psi_m( \mathbf x) -t_m=0, \quad \quad m=1,2, \ldots, M,\label{cvx1}\\
&\phi_m ( \mathbf x)-t_m=0, \quad \quad m=1,2,\ldots, M.\label{cvx2}
\end{align}
\end{subequations}
Since (\ref{cvx1}) and (\ref{cvx2}) are both convex, this problem is the desired convex programming equivalent of the non-convex one.
{
\bibliographystyle{IEEEtranTCOM}
\bibliography{mybibfile}%

\begin{thebibliography}{10}
\baselineskip 12pt
\providecommand{\url}[1]{#1}
\csname url@samestyle\endcsname
\providecommand{\newblock}{\relax}
\providecommand{\bibinfo}[2]{#2}
\providecommand{\BIBentrySTDinterwordspacing}{\spaceskip=0pt\relax}
\providecommand{\BIBentryALTinterwordstretchfactor}{4}
\providecommand{\BIBentryALTinterwordspacing}{\spaceskip=\fontdimen2\font plus
\BIBentryALTinterwordstretchfactor\fontdimen3\font minus \fontdimen4\font\relax}
\providecommand{\BIBforeignlanguage}[2]{{%
\expandafter\ifx\csname l@#1\endcsname\relax
\typeout{** WARNING: IEEEtran.bst: No hyphenation pattern has been}%
\typeout{** loaded for the language `#1'. Using the pattern for}%
\typeout{** the default language instead.}%
\else
\language=\csname l@#1\endcsname
\fi
#2}}
\providecommand{\BIBdecl}{\relax}
\BIBdecl

\bibitem{2020survey}
S.~Fattah \emph{et~al.}, ``A survey on underwater wireless sensor networks: Requirements, taxonomy, recent advances, and open research challenges,'' \emph{Sensors}, vol.~20, no.~18, p. 5393, 2020.

\bibitem{wadaa2005training}
A.~Wadaa \emph{et~al.}, ``Training a wireless sensor network,'' \emph{Mobile Networks and Applications}, vol.~10, no. 1-2, pp. 151--168, 2005.

\bibitem{wang2015energy}
K.~Wang \emph{et~al.}, ``An energy-efficient reliable data transmission scheme for complex environmental monitoring in underwater acoustic sensor networks,'' \emph{IEEE Sensors Journal}, vol.~16, no.~11, pp. 4051--4062, 2015.

\bibitem{zhuo2020auv}
X.~Zhuo \emph{et~al.}, ``Auv-aided energy-efficient data collection in underwater acoustic sensor networks,'' \emph{IEEE Internet of Things Journal}, 2020.

\bibitem{yan2018energy}
J.~Yan \emph{et~al.}, ``Energy-efficient data collection over auv-assisted underwater acoustic sensor network,'' \emph{IEEE Systems Journal}, vol.~12, no.~4, pp. 3519--3530, 2018.

\bibitem{javaid2015efficient}
N.~Javaid \emph{et~al.}, ``An efficient data-gathering routing protocol for underwater wireless sensor networks,'' \emph{Sensors}, vol.~15, no.~11, pp. 29\,149--29\,181, 2015.

\bibitem{power_control}
H.~U. Yildiz, V.~C. Gungor, and B.~Tavli, ``Packet size optimization for lifetime maximization in underwater acoustic sensor networks,'' \emph{IEEE Transactions on Industrial Informatics}, vol.~15, no.~2, pp. 719--729, 2018.

\bibitem{power_control1}
P.~Anjangi and M.~Chitre, ``Scheduling algorithm with transmission power control for random underwater acoustic networks,'' in \emph{OCEANS 2015-Genova}.\hskip 1em plus 0.5em minus 0.4em\relax IEEE, 2015, pp. 1--8.

\bibitem{mac1}
A.~Roy and N.~Sarma, ``Rpcp-mac: Receiver preambling with channel polling mac protocol for underwater wireless sensor networks,'' \emph{International Journal of Communication Systems}, vol.~33, no.~9, p. e4383, 2020.

\bibitem{mac2}
H.~Wang \emph{et~al.}, ``An effective scheduling algorithm for coverage control in underwater acoustic sensor network,'' \emph{Sensors}, vol.~18, no.~8, p. 2512, 2018.

\bibitem{EH1}
H.~E. Erdem, H.~U. Yildiz, and V.~C. Gungor, ``On the lifetime of compressive sensing based energy harvesting in underwater sensor networks,'' \emph{IEEE Sensors Journal}, vol.~19, no.~12, pp. 4680--4687, 2019.

\bibitem{EH2}
L.~Jing \emph{et~al.}, ``Energy management and power allocation for underwater acoustic sensor network,'' \emph{IEEE Sensors Journal}, vol.~17, no.~19, pp. 6451--6462, 2017.

\bibitem{ahmed2016optimized}
T.~Ahmed \emph{et~al.}, ``Optimized depth-based routing protocol for underwater wireless sensor networks,'' in \emph{Open Source Systems \& Technologies (ICOSST), 2016 International Conference on}.\hskip 1em plus 0.5em minus 0.4em\relax IEEE, 2016, pp. 147--150.

\bibitem{wadud2019energy}
Z.~Wadud \emph{et~al.}, ``An energy balanced efficient and reliable routing protocol for underwater wireless sensor networks,'' \emph{IEEE Access}, vol.~7, pp. 175\,980--175\,999, 2019.

\bibitem{felemban2015underwater}
E.~Felemban \emph{et~al.}, ``Underwater sensor network applications: A comprehensive survey,'' \emph{International Journal of Distributed Sensor Networks}, vol.~11, no.~11, p. 896832, 2015.

\bibitem{das2017enhancement}
B.~Das, E.~S. Mishra, and S.~K. Sethi, ``Enhancement of lifetime of acoustic sensor using pegasis algorithm in uasn,'' \emph{International Journal of Electronics, Electrical and Computational System (Mc Graw Hill Publication)}, vol.~6, no.~9, pp. 534--545, 2017.

\bibitem{liu2017optimal}
L.~Liu \emph{et~al.}, ``Optimal relay node placement and flow allocation in underwater acoustic sensor networks,'' \emph{IEEE Transactions on Communications}, vol.~65, no.~5, pp. 2141--2152, 2017.

\bibitem{mohammadi2018new}
Z.~Mohammadi \emph{et~al.}, ``A new optimization algorithm for relay node setting in underwater acoustic sensor networks,'' in \emph{2018 3rd Conference on Swarm Intelligence and Evolutionary Computation (CSIEC)}.\hskip 1em plus 0.5em minus 0.4em\relax IEEE, 2018, pp. 1--5.

\bibitem{mohammadi2020increasing}
------, ``Increasing the lifetime of underwater acoustic sensor networks: Difference convex approach,'' \emph{IEEE Systems Journal}, 2020.

\bibitem{mohammadi2022modified}
Z.~Mohammadi, ``Modified relay node placement in dense 3d underwater acoustic sensor networks,'' in \emph{2022 9th Iranian Joint Congress on Fuzzy and Intelligent Systems (CFIS)}.\hskip 1em plus 0.5em minus 0.4em\relax IEEE, 2022, pp. 1--4.

\bibitem{akyildiz2005underwater}
I.~F. Akyildiz, D.~Pompili, and T.~Melodia, ``Underwater acoustic sensor networks: research challenges,'' \emph{Ad hoc networks}, vol.~3, no.~3, pp. 257--279, 2005.

\bibitem{urick1975}
R.~J. Urick, ``Principles of underwater sound-2,'' 1975.

\bibitem{wang2016energy}
K.~Wang \emph{et~al.}, ``An energy-efficient reliable data transmission scheme for complex environmental monitoring in underwater acoustic sensor networks,'' \emph{IEEE Sensors Journal}, vol.~16, no.~11, pp. 4051--4062, 2016.

\bibitem{thorp1967analytic}
W.~H. Thorp, ``Analytic description of the low-frequency attenuation coefficient,'' \emph{The Journal of the Acoustical Society of America}, vol.~42, no.~1, pp. 270--270, 1967.

\bibitem{cao2015balance}
J.~Cao, J.~Dou, and S.~Dong, ``Balance transmission mechanism in underwater acoustic sensor networks,'' \emph{International Journal of Distributed Sensor Networks}, vol.~11, no.~3, p. 429340, 2015.

\bibitem{akbar2016efficient}
M.~Akbar \emph{et~al.}, ``Efficient data gathering in 3d linear underwater wireless sensor networks using sink mobility,'' \emph{Sensors}, vol.~16, no.~3, p. 404, 2016.

\bibitem{boyd2004convex}
S.~Boyd and L.~Vandenberghe, \emph{Convex optimization}.\hskip 1em plus 0.5em minus 0.4em\relax Cambridge university press, 2004.

\bibitem{tuy1998convex}
H.~Tuy \emph{et~al.}, \emph{Convex analysis and global optimization}.\hskip 1em plus 0.5em minus 0.4em\relax Springer, 1998.

\bibitem{silverman1989essential}
R.~A. Silverman, \emph{Essential calculus with applications}.\hskip 1em plus 0.5em minus 0.4em\relax Courier Corporation, 1989, vol.~75.

\bibitem{soleimanpour2018low}
M.~Soleimanpour-moghadam \emph{et~al.}, ``Low complexity green cooperative cognitive radio network with superior performance,'' \emph{IEEE Systems Journal}, no.~99, pp. 1--12, 2018.

\bibitem{stepapprox}
F.~Rinaldi, F.~Schoen, and M.~Sciandrone, ``Concave programming for minimizing the zero-norm over polyhedral sets,'' \emph{Computational Optimization and Applications}, vol.~46, no.~3, pp. 467--486, 2010.

\bibitem{alsalih2010placement}
W.~Alsalih, H.~Hassanein, and S.~Akl, ``Placement of multiple mobile data collectors in wireless sensor networks,'' \emph{Ad Hoc Networks}, vol.~8, no.~4, pp. 378--390, 2010.

\end{thebibliography}
}

\end{document}